\def \VersionAuthorforArXiV {}
\ifdefined\VersionAuthorforArXiV
	\newcommand{\arXivVersion}[1]{#1}
	\newcommand{\FinalACMVersion}[1]{}
\else
	\newcommand{\arXivVersion}[1]{}
	\newcommand{\FinalACMVersion}[1]{#1}
\fi

\def \VersionLong {}
\ifdefined \VersionLong
	\newcommand{\LongVersion}[1]{\ifdefined\VersionWithComments{\color{red!40!black}#1}\else#1\fi}
	\newcommand{\ShortVersion}[1]{\ifdefined\VersionWithComments{\color{black!40}#1}\fi}
\else
	\newcommand{\LongVersion}[1]{\ifdefined\VersionWithComments{\color{black!60}#1}\fi}
	\newcommand{\ShortVersion}[1]{\ifdefined\VersionWithComments{\color{red!40!black}#1}\else#1\fi}
\fi

\newcommand{\LongVersionTable}[1]{}
\newcommand{\ShortVersionTable}[1]{#1}
\newcommand{\nbLoC}[1]{\cellcolor{gray!20}}

\ifdefined\VersionAuthorforArXiV
	\documentclass[]{article} %
\else
	\PassOptionsToPackage{svgnames,table}{xcolor}
	\documentclass[acmsmall,
	natbib=false%
	]{acmart}
\fi
\ifdefined\VersionAuthorforArXiV
\else
	\setcopyright{acmcopyright}
	\copyrightyear{2018}
	\acmYear{2018}
	\acmDOI{10.1145/1122445.1122456}

	\acmJournal{TOSEM}
	\acmVolume{37}
	\acmNumber{4}
	\acmArticle{111}
	\acmMonth{8}
\fi

\usepackage[utf8]{inputenc}
\usepackage[english]{babel}

\usepackage[ruled,vlined,linesnumbered]{algorithm2e}
	\SetKwInOut{Input}{input}
	\SetKwInOut{Output}{output}

\usepackage{subcaption}

\usepackage{paralist} %
\usepackage{adjustbox}

\newenvironment{ienumeration}
	{\ifdefined\VersionLong\begin{enumerate}\else\begin{inparaenum}[\itshape i\upshape)]\fi}
	{\ifdefined\VersionLong\end{enumerate}\else\end{inparaenum}\fi}

\newenvironment{oneenumeration}
	{\ifdefined\VersionLong\begin{enumerate}\else\begin{inparaenum}[1)]\fi}
	{\ifdefined\VersionLong\end{enumerate}\else\end{inparaenum}\fi}

\arXivVersion{%
	\usepackage{amsmath} %
	\usepackage{amssymb} %
}

\ifdefined\VersionWithComments
	\usepackage{draftwatermark}
	\SetWatermarkText{draft}
	\SetWatermarkScale{2}
	\SetWatermarkColor[gray]{0.9}
\fi

\arXivVersion{%
	\usepackage[svgnames,table]{xcolor}
	\definecolor{darkblue}{rgb}{0, 0, 0.7}

	\usepackage[
			pdfauthor={Étienne André, Didier Lime, Dylan Marinho and Sun Jun},%
			pdftitle={Guaranteeing Timed Opacity using Parametric Timed Model Checking},
			breaklinks  = true,
			colorlinks  = true,
			citecolor   = blue!50!blue,
			linkcolor   = darkblue,
			urlcolor    = blue!50!green,
		]{hyperref}
}

\usepackage[capitalise,english,nameinlink]{cleveref} %
\crefname{line}{\text{line}}{\text{lines}} %

\usepackage{listings}

\definecolor{mygreen}{rgb}{0,0.6,0}
\definecolor{mygray}{rgb}{0.5,0.5,0.5}
\definecolor{mymauve}{rgb}{0.58,0,0.82}

\newcommand{\defProblem}[3]
{%
	\noindent\fcolorbox{black}{blue!15}{
	\begin{minipage}{.95\columnwidth}
		\textbf{#1 Problem:}\\
		\textsc{Input}: #2\\
		\textsc{Problem}: #3
	\end{minipage}
}

	\smallskip

}

\usepackage{tikz}
\usepackage{pgfplots}
\pgfplotsset{compat=1.15}
\usepackage{mathrsfs}
\usetikzlibrary{arrows}
\usetikzlibrary{patterns}
\usetikzlibrary{arrows,automata,positioning,shapes} %
\tikzstyle{every node}=[initial text=]
\tikzstyle{location}=[rectangle, rounded corners, minimum size=12pt, draw=black, fill=blue!10, inner sep=2pt]
\tikzstyle{invariant}=[draw=black, dotted, inner sep=1pt, node distance=0] %
\tikzstyle{final}=[double, fill=blue!50]

\tikzstyle{urgent}=[fill=yellow, thick, dotted] %
\tikzstyle{private}=[fill=red,thick]

\definecolor{coloract}{rgb}{0.50, 0.70, 0.30}
\definecolor{colorclock}{rgb}{0.4, 0.4, 1}
\definecolor{colordisc}{rgb}{1, 0, 1}
\definecolor{colorloc}{rgb}{0.4, 0.4, 0.65}
\definecolor{colorparam}{rgb}{1, 0.6, 0.0}

\definecolor{loccolor1}{rgb}{1, 0.3, 0.3}
\definecolor{loccolor2}{rgb}{0.3, 1, 0.3}
\definecolor{loccolor3}{rgb}{0.3, 0.3, 1}
\definecolor{loccolor4}{rgb}{1, 0.3, 1}
\definecolor{loccolor5}{rgb}{1, 1, 0.3}
\definecolor{loccolor6}{rgb}{0.3, 1, 1}
\definecolor{loccolor7}{rgb}{0.9, 0.6, 0.2}
\definecolor{loccolor8}{rgb}{0.7, 0.4, 1}
\definecolor{loccolor9}{rgb}{0.5, 1, 0.75}
\definecolor{loccolor10}{rgb}{0.8, 0.7, 0.6}
\definecolor{loccolor11}{rgb}{0.6, 0.7, 0.8}
\definecolor{loccolor12}{rgb}{0.2, 0.5, 0.9}
\definecolor{loccolor13}{rgb}{0.5, 0.9, 0.2}
\definecolor{loccolor14}{rgb}{0.9, 0.2, 0.5}
\definecolor{loccolor15}{rgb}{0.7, 0.7, 0.7}
\definecolor{loccolor16}{rgb}{0.8, 0.8, 0.5}

\newcommand{\styleact}[1]{\ensuremath{\textcolor{coloract}{{#1}}}}
\newcommand{\styleclock}[1]{\ensuremath{\textcolor{colorclock}{{#1}}}}
\newcommand{\styledisc}[1]{\ensuremath{\textcolor{colordisc}{\mathrm{#1}}}}
\newcommand{\styleloc}[1]{\ensuremath{\mathrm{#1}}}
\newcommand{\styleparam}[1]{\ensuremath{\textcolor{colorparam}{{#1}}}}

\newcommand{\clockx}{\ensuremath{\styleclock{x}}}

\newcommand{\clockcl}{\styleclock{\mathit{cl}}}

\newcommand{\stylecode}[1]{\textcolor{colorloc}{\texttt{#1}}}

\newcommand{\stylebench}[1]{\textcolor{colorloc}{\texttt{#1}}}

\newcommand{\cellHeader}[0]{\cellcolor{blue!20}\bfseries}
\newcommand{\rowHeader}{\rowcolor{blue!20}\bfseries}
\newcommand{\cellYes}{\cellcolor{green!20}\textbf{$\surd$}}
\newcommand{\cellNo}{\cellcolor{red!20}\textbf{$\times$}}
\newcommand{\cellFixable}{\cellcolor{orange!20}\textbf{$(\times)$}}

\newcommand{\cellKall}{\cellcolor{blue!20}\textbf{$\KTrue$}}
\newcommand{\cellKnone}{\cellcolor{red!20}\textbf{$\KFalse$}}
\newcommand{\cellKsome}{\cellcolor{green!20}\textbf{$K$}}

\newcommand{\init}{_0}

\newcommand{\A}{\ensuremath{\mathcal{A}}}
\newcommand{\Azeroinf}{\ensuremath{\A_{0,\infty}}}
\newcommand{\Actions}{\Sigma}
\newcommand{\action}{\ensuremath{a}}
\newcommand{\actionEnd}{\ensuremath{\styleact{\mathrm{finish}}}}
\newcommand{\ActionsIndices}{\zeta}
\newcommand{\assign}{\leftarrow}
\newcommand{\bflag}{\ensuremath{b}} %

\newcommand{\BTrue}{\text{true}}
\newcommand{\BFalse}{\text{false}}
\newcommand{\Constr}{C}
\newcommand{\class}[1]{\ensuremath{\left[#1\right]}}
\newcommand{\Clock}{\mathbb{X}} %
\newcommand{\ClockCard}{H} %
\newcommand{\clock}{x} %
\newcommand{\clockabs}{\ensuremath{x_\mathit{abs}}} %
\newcommand{\clockval}{\mu} %
\newcommand{\ClocksZero}{\vec{0}}
\newcommand{\compOp}{\bowtie}

\newcommand{\compOpLeq}{\triangleleft}
\newcommand{\CTrue}{\mathbf{true}}

\newcommand{\duration}{\ensuremath{\mathit{dur}}}
\newcommand{\edge}{e}
\newcommand{\Edges}{E}
\newcommand{\integralp}[1]{\ensuremath{\lfloor#1\rfloor}}
\newcommand{\equivalent}{\ensuremath{\approx}}
\newcommand{\fract}[1]{\ensuremath{\text{fract}(#1)}}
\newcommand{\longuefleche}[1]{\stackrel{#1}{\longrightarrow}}
\newcommand{\longueflecheRel}[1]{\stackrel{#1}{\mapsto}}

\newcommand{\flecheRel}{{\rightarrow}}

\newcommand{\guard}{g}

\newcommand{\invariant}{I}
\newcommand{\K}{K}
\newcommand{\KTrue}{\top}
\newcommand{\KFalse}{\bot}

\newcommand{\loc}{\ensuremath{\ell}} %
\newcommand{\locinit}{\loc\init}
\newcommand{\Loc}{L} %
\newcommand{\locfinal}{\ensuremath{\loc_f}}
\newcommand{\locpriv}{\ensuremath{\loc_{\mathit{priv}}}}
\newcommand{\locpub}{\ensuremath{\loc_{\mathit{pub}}}}
\newcommand{\locTarget}{\ensuremath{\loc_{T}}}
\newcommand{\lterm}{\mathit{lt}}
\newcommand{\Param}{\mathbb{P}} %
\newcommand{\param}{p} %
\newcommand{\paramabs}{\ensuremath{\param_\mathit{abs}}} %
\newcommand{\ParamCard}{M} %
\newcommand{\pval}{v} %
\newcommand{\PZG}{\ensuremath{\mathcal{PZG}}} %
\newcommand{\R}{{\mathbb{R}}}
\newcommand{\region}{\ensuremath{r}}
\newcommand{\Regions}{{\mathcal{R}}}
\newcommand{\RegionGraph}{{\mathcal{RG}}}
\newcommand{\regionEdges}{{\mathcal{F}}}
\newcommand{\Rgeqzero}{\R_{\geq 0}}

\newcommand{\sinit}{s\init} %
\newcommand{\somelocs}{\ensuremath{\Loc_T}} %
\newcommand{\concstate}{\ensuremath{s}} %
\newcommand{\States}{S} %
\newcommand{\Succ}{\mathsf{Succ}}
\newcommand{\timelapse}[1]{#1^\nearrow}
\newcommand{\Times}{\ensuremath{D}}
\newcommand{\TTS}{\ensuremath{T}}
\newcommand{\varrun}{\rho} %

\newcommand{\setn}{{\mathbb N}}
\newcommand{\setq}{{\mathbb Q}}
\newcommand{\setqplus}{\setq_{+}} %
\newcommand{\setr}{\ensuremath{\mathbb R}}
\newcommand{\setrplus}{\ensuremath{\setr_{+}}} %
\newcommand{\setz}{{\mathbb Z}}

\newcommand{\PrivDurReach}[3]{\ensuremath{\mathit{DReach}^{#1}_{#2}(#3)}}
\newcommand{\PubDurReach}[3]{\ensuremath{\mathit{DReach}^{#1}_{\neg #2}(#3)}}

\newcommand{\set}[1]{\ensuremath{\left\{ #1 \right\}}}
\newcommand{\Set}[2]{\ensuremath{\left\{ #1 \ | \ #2 \right\}}}

\newcommand{\durRegions}[2]{\ensuremath{\lambda_{#1,#2}}}

\newcommand{\styleSymbStatesSet}[1]{\ensuremath{\mathbf{#1}}}

\newcommand{\symbstate}{\ensuremath{\styleSymbStatesSet{s}}} %
\newcommand{\SymbState}{\ensuremath{\styleSymbStatesSet{S}}} %
\newcommand{\symbstateinit}{\symbstate\init} %
\newcommand{\symbtrans}{{\Rightarrow}} %

\newcommand{\resets}{R}
\newcommand{\projectP}[1]{\ensuremath{#1{\downarrow_{\Param}}}}
\newcommand{\reset}[2]{\ensuremath{[#1]_{#2}}}
\newcommand{\valuate}[2]{\ensuremath{#2(#1)}}
\newcommand{\stylealgo}[1]{\ensuremath{\textsf{#1}}}
\newcommand{\Copy}{\stylealgo{Copy}}
\newcommand{\EFsynth}{\stylealgo{EFsynth}}
\newcommand{\Enrich}{\stylealgo{Enrich}}
\newcommand{\SynthOp}{\stylealgo{SynthOp}}

\arXivVersion{%
	\usepackage{amsthm}
	\theoremstyle{plain}
	\newtheorem{lemma}{Lemma}
	\newtheorem{proposition}{Proposition}
	\newtheorem{theorem}{Theorem}

	\theoremstyle{definition}
	\newtheorem{definition}{Definition}
	\newtheorem{example}{Example}

}

\theoremstyle{remark}
\newtheorem{remark}{Remark}

\ifdefined\VersionWithComments
	\usepackage[colorinlistoftodos,textsize=footnotesize]{todonotes}
\else
	\usepackage[disable]{todonotes}
\fi

\ifdefined \VersionWithComments
	\newcommand{\todoinline}[1]{\mbox{}{\color{red}{\textbf{TODO}\ifx#1\\\else:\ \fi #1}}} %
\else
	\newcommand{\todoinline}[1]{}
\fi

\usepackage{verbatim} %

\ifdefined\VersionAuthorforArXiV
	\usepackage[backend=biber,backref=true,style=alphabetic,url=false,doi=true,defernumbers=true,sorting=anyt,maxnames=99]{biblatex} %
	\renewbibmacro*{doi+eprint+url}{%
		\iftoggle{bbx:doi}
			{\color{black!40}\footnotesize\printfield{doi}}
			{}%
		\newunit\newblock
		\iftoggle{bbx:eprint}
			{\usebibmacro{eprint}}
			{}%
		\newunit\newblock
		\iftoggle{bbx:url}
			{\usebibmacro{url+urldate}}
			{}%
	}

\else
	\usepackage[style=ACM-Reference-Format,backend=bibtex,sorting=none]{biblatex}
\fi

\newcommand{\imitator}{\textsf{IMITATOR}}
\newcommand{\SpaceEx}{\textsc{SpaceEx}}
\newcommand{\uppaal}{\textsc{Uppaal}}

\ifdefined \VersionWithComments
 	\definecolor{colorok}{RGB}{80,80,150}
\else
	\definecolor{colorok}{RGB}{0,0,0}
\fi

\newcommand{\eg}{\textcolor{colorok}{e.\,g.,}\xspace}
\newcommand{\ie}{\textcolor{colorok}{i.\,e.,}\xspace}
\newcommand{\st}{\textcolor{colorok}{s.t.}\xspace}

\newcommand{\wrt}{{w.r.t.}\xspace} %

\newcommand{\ouracks}{%
	We would like to thank Sudipta Chattopadhyay for helpful suggestions, Nicolas Markey for an interesting discussion on the proof of \cref{proposition:ET-opacity-computation}, Jiaying Li for his help with preliminary model conversion,
	and an anonymous reviewer of ATVA~2019 for suggesting \cref{remark:R2}.
	This work is partially supported
			by the ANR national research program PACS (ANR-14-CE28-0002),
			by the ANR-NRF research program ProMiS (ANR-19-CE25-0015),
			and
		by ERATO HASUO Metamathematics for Systems Design Project (No.\ JPMJER1603), JST.
}

\newcommand{\ourkeywords}{opacity, timed automata, \imitator{}, parameter synthesis}

\newcommand{\ourabstract}{%
	Information leakage can have dramatic consequences on systems security.
	Among harmful information leaks, the timing information leakage occurs whenever an attacker successfully deduces confidential internal information.
	In this work, we consider that the attacker has access (only) to the system execution time.
	We address the following timed opacity problem:
		given a timed system, a private location and a final location,
			synthesize the execution times from the initial location to the final location for which one cannot deduce whether the system went through the private location.
		We also consider the full timed opacity problem, asking whether the system is opaque for all execution times.
	We show that these problems are decidable for timed automata (TAs) but become undecidable when one adds parameters, yielding parametric timed automata (PTAs).
	We identify a subclass with some decidability results.
	We then devise an algorithm for synthesizing PTAs parameter valuations guaranteeing that the resulting TA is opaque.
	We finally show that our method can also apply to program analysis.
}

\arXivVersion{
\makeatletter
\def\orcidID#1{\,\smash{\href{https://orcid.org/#1}{\protect\raisebox{%
	+1.25pt%
}{\protect\includegraphics{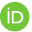}}}}}
\makeatother
}

\FinalACMVersion{%
}

\begin{document}
\sloppy

\title{Guaranteeing Timed Opacity using Parametric Timed Model Checking\arXivVersion{\thanks{%
	This is the author version of the manuscript of the same name published in ACM Transactions on Software Engineering and Methodology (ToSEM).
	The final version is available at \href{https://dl.acm.org/journal/tosem}{\nolinkurl{acm.org}}.
	\ouracks{}
	}}
}

\ifdefined\VersionAuthorforArXiV
	\author{\'Etienne Andr\'e\orcidID{0000-0001-8473-9555}$^1$,
	Didier Lime\orcidID{0000-0001-9429-7586}$^2$,
	Dylan Marinho\orcidID{0000-0002-2548-6196}$^1$
	and
	Jun Sun\orcidID{0000-0002-3545-1392}$^3$
	\\
		$^1$Université de Lorraine, CNRS, Inria, LORIA, Nancy, France\\
		$^2$École Centrale de Nantes, LS2N, UMR CNRS 6004, Nantes\\
		$^3$School of Information Systems, Singapore Management University\\
	}
	\date{}
\else
	\author{\'Etienne Andr\'e}
	\orcid{0000-0001-8473-9555}
	\affiliation{%
	\institution{Université de Lorraine, CNRS, Inria, LORIA, F-54000 Nancy}
	\country{France}
	}

	\author{Didier Lime}
	\orcid{0000-0001-9429-7586}
	\affiliation{%
	\institution{École Centrale de Nantes, LS2N, UMR CNRS 6004, Nantes}
	\country{France}}

	\author{Dylan Marinho}
	\orcid{0000-0002-2548-6196}
	\affiliation{%
	\institution{Université de Lorraine, CNRS, Inria, LORIA, F-54000 Nancy}
	\country{France}
	}

	\author{Jun Sun}
	\orcid{0000-0002-3545-1392}
	\affiliation{%
	\institution{School of Information Systems, Singapore Management University}
	\country{Singapore}
	}

\begin{abstract}
	\ourabstract{}
\end{abstract}



\fi
\maketitle

\arXivVersion{%
	\newcommand{\keywords}[1]
	{%
		\small\textbf{\textit{Keywords---}} #1
	}

	\begin{abstract}
		\ourabstract{}
	\end{abstract}

}

\todo{This is the version with comments. To disable comments, comment out line~3 in the \LaTeX{} source.}

\ifdefined\VersionWithComments
	\tableofcontents{}
\fi
\section{Introduction}\label{section:introduction}

Timed systems often combine hard real-time constraints with other complications such as concurrency.
Information leakage can have dramatic consequences on the security of such systems.
Among harmful information leaks, the \emph{timing information leakage} is the ability for an attacker to deduce internal information depending on timing information.
In this work, we focus on timing leakage through the total execution time, \ie{} when a system works as an almost black-box and the ability of the attacker is limited to know the model and observe the total execution time.

\paragraph{Opacity}
In its most general form on partially observed labeled transitions systems, given a set of runs that reveal a secret (\eg{} they perform a secret action or visit a secret state), \emph{opacity} states that if there exists a run of the system that reveals the secret (\ie{} belongs to the given secret set), there exists another run, with the same observation, that does not reveal that secret~\cite{BKMR08}.
This secret is completely generic and, depending on its actual definition, properties and their decidability can differ.

In our setting, we define a form of opacity in which the observation is only \emph{the time to reach a designated location}.

\paragraph{Contributions for TAs}
In this work, we consider the setting of timed automata (TAs), which is a popular extension of finite-state automata with clocks~\cite{AD94}.
We consider the following version of timed opacity:
	given a TA, a private location denoting the execution of some secret behavior and a final location denoting the completion of the execution,
	the TA is opaque for a given execution time $d$ (\ie{} the time of a run from the initial location to the final location) if there exist two runs of duration~$d$ from the initial location to the final location, one going through the private location, and another run \emph{not} going through the private location.
That is, for this particular execution time, the system is opaque if one cannot deduce whether the system went through the private location.
Such a notion of timed opacity can be used to capture many interesting security problems: for instance, it is possible to deduce whether a secret satisfies a certain condition based on whether a certain branch is visited or not.

To be explicit, the attacker knows a TA model of the system, and can observe the execution time from the system start until it reaches some particular final location.
No other actions can be observed.
Then, the system is timed opaque if the attacker cannot deduce whether the system has visited some particular private location.
From a higher-level point of view, this means that the attacker cannot deduce some private information, such as whether some location has been visited, or whether some branch of a given program was visited, by only observing the execution time.
In practice, this corresponds to a setting where the attacker may interact with some computational process on a remote machine (\eg{} a server) and receives the responses only at the end of the process (\eg{} a server message is received).
\label{newtext:attacker}

We consider two problems based on this notion of timed opacity:
\begin{enumerate}
	\item a computation problem: the computation of the set of possible execution times for which the system is timed opaque; and
	\item a decision problem: whether the TA is timed opaque for all execution times (referred to as full timed opacity).
\end{enumerate}

We first prove that these problems can be effectively solved for TAs.
We implement our procedure and apply it to a set of benchmarks containing notably a set of Java programs known for their (absence of) timing information leakage.

\paragraph{Contributions for parametric TAs}
As a second setting, we consider a higher-level version of these problems by allowing (internal) timing parameters in the system, which can model uncertainty or unknown constants at early design stage.
The setting becomes parametric timed automata (PTAs)~\cite{AHV93}.

On the theoretical side, we answer an existential parametric version of the two aforementioned problems, that is, the existence of (at least) one parameter valuation for which the TA is (fully) timed opaque.
Although we show that these problems are in general undecidable, we exhibit a subclass with some decidability results.

Then, we address a practical problem:
	given a timed system with timing parameters, a private location and a final location,
		synthesize the timing parameters and the execution times for which one cannot deduce whether the system went through the private location.
We devise a general procedure not guaranteed to terminate, but that behaves well on examples from the literature.

\paragraph{Summary of the contributions}
To sum up, this manuscript proposes the following contributions:
\begin{enumerate}
	\item a notion of timed opacity, and a notion of full timed opacity for TAs;
	\item a procedure to solve the timed opacity computation problem for TAs, and a procedure to answer the full timed opacity decision problem for TAs;
	\item a study of two theoretical decision problems extending the two aforementioned problems to the parametric setting, and exhibition of a decidable subclass;
	\item a practical algorithm to synthesize parameter valuations and execution times for which the TA is guaranteed to be opaque;
	\item a set of experiments on a set of benchmarks, including PTAs translations from Java programs.
\end{enumerate}

This manuscript is an extension of~\cite{AS19} with the following improvements.
\begin{itemize}
	\item We provide all proofs of the results published in~\cite{AS19}.

	\item We %
		extend the theoretical part, by considering not one problem (as in~\cite{AS19}) but two versions (timed opacity \wrt{} a set of execution times, and full timed opacity), both for TAs and PTAs (including the subclass of L/U-PTAs).

	\item We propose a more elegant proof of \cref{proposition:ET-opacity-computation} (formerly \cite[Proposition~1]{AS19}), based on RA arithmetic~\cite{W99}.

	\item On the practical side, we give hints to extend our construction to a richer framework (\cref{ss:richer}).
\end{itemize}

\paragraph*{Outline}
After reviewing related works in \cref{section:related},
\cref{section:preliminaries} recalls necessary concepts and \cref{section:problem} introduces the problem.
\cref{section:TA} addresses timed opacity for timed automata.
We then address the parametric version of timed opacity, with theory studied in \cref{section:theory,section:theoryFullOpacity},
algorithmic in~\cref{section:synthesis}
and experiments in \cref{section:experiments}.
\cref{section:conclusion} concludes the paper.

\section{Related works}\label{section:related}

\paragraph{Opacity and timed automata}
This work is closely related to the line of work on defining and analyzing information flow in timed automata.
It is well-known (see \eg{} \cite{Kocher96,FS00,BB07,KPJJ13,BCLR15}) that time is a potential attack vector against secure systems.
That is, it is possible that a non-interferent (secure) system can become interferent (insecure) when timing constraints are added~\cite{GMR07}.

In \emph{non-interference}, actions are partitioned into two levels of privilege, \emph{high} and \emph{low}, and we require that the system in which high-level actions are removed is equivalent to the system in which they are hidden (\ie{} replaced by an unobservable action).
Different equivalences lead to different flavors of non-interference. 
In~\cite{BDST02,BT03}, a first notion of \emph{timed} non-interference is proposed for TAs.
This notion is extended to PTAs in~\cite{AK20}, with a semi-algorithm.

In~\cite{GMR07}, Gardey \emph{et al.}\ define timed strong non-deterministic non-interference (SNNI) based on timed language equivalence between the automaton with hidden low-level actions and the automaton with removed low-level actions. Furthermore, they show that the problem of determining whether a timed automaton satisfies SNNI is undecidable. In contrast, timed cosimulation-based SNNI, timed bisimulation-based SNNI and timed state SNNI are decidable.
Classical SNNI is the one corresponding to the equality of the languages of the two systems.
As such it is clearly a special case of opacity in which the secret runs are those containing a high-level action~\cite{BKMR08}.
Other equivalence relations (namely (timed) cosimulation, (timed) bisimulation, sets of states) are not as easily relatable to opacity.
No implementation is provided in~\cite{GMR07}.

In~\cite{Cassez09}, it is proved that it is undecidable whether a TA is opaque, for the following definition of opacity: the system is opaque if an attacker cannot deduce whether some set of actions was performed, by only observing a given set of observable actions together with their timestamp.
This problem is proved undecidable even for the restricted class of event-recording automata~\cite{AFH99}, which is a subclass of TAs.
No implementation nor procedure is provided.
In contrast, our definition of opacity is decidable for TAs, notably because in our setting the attacker power is more restricted (they can only observe the ``execution time''); in addition, our definition of opacity has some practical relevance nonetheless, when an attacker is able to interact remotely with the system under attack, and is therefore able to measure the response time.

In~\cite{AEYM21}, the authors consider a \emph{time-bound} opacity, where the attacker has to disclose the secret before an upper bound, using a partial observability.
The authors prove that this problem is decidable for~TAs.
A construction and an algorithm are also provided to solve it; a case study is verified using \SpaceEx{}~\cite{FLDCRLRGDM11}.
In contrast, our definition of opacity only assumes observation of the execution time, does not assume any time-bounded setting, and our most general problem is parametric.

In~\cite{NNV17}, the authors propose a type system dealing with non-determinism and (continuous) real-time, the adequacy of which is ensured using non-interference.
We share the common formalism of TAs; however, we mainly focus on leakage as execution time, and we \emph{synthesize} internal parts of the system (clock guards), in contrast to~\cite{NNV17} where the system is fixed.

In~\cite{VNN18}, Vasilikos \emph{et al.}\ define the security of timed automata in term of information flow using a bisimulation relation over a set of observable nodes and develop an algorithm for deriving a sound constraint for satisfying the information flow property locally based on relevant transitions.
\label{newtext:Vasilikos}

In~\cite{GSB18}, Gerking \emph{et al.}\ study non-interference properties with input, high and low actions and provide a resolution method reducing a secure behavior to an unreachability construction. The proof-of-concept consists in the exhibition of a test automaton with a dedicated location that indicates violations of noninterference whenever it is reachable during execution.
Then, \uppaal{}~\cite{LPY97} is used to obtain the answer.

In~\cite{BCLR15}, Benattar \emph{et al.}\ study the control synthesis problem of timed automata for SNNI.
That is, given a timed automaton, they propose a method to automatically generate a (largest) sub-system such that it is non-interferent, if possible.
Different from the above-mentioned work, our work considers \emph{parametric} timed automata, \ie{} timed systems with unknown design parameters, and focuses on synthesizing parameter valuations which guarantee information flow property.
Compared to~\cite{BCLR15}, our approach is more realistic as it does not require change of program structure.
Rather, our result provides guidelines on how to choose the timing parameters (\eg{} how long to wait after certain program statements) for avoiding information leakage.
In~\cite{WZ17,WZA18}, Wang \emph{et al.}\ investigate interesting opacity problems for real-time automata.
These works come with a dedicated Python implementation.
Although their definition shares similarities with ours, real-time automata are a severely restricted formalism compared to TAs.
Indeed, timed aspects are only considered by interval restrictions over the total elapsed time along transitions.
Real-time automata can be seen as a subclass of TAs with a single clock, reset at each transition.
Also, parameters are not considered in their work.

To the best of our knowledge, our approach is the first work on parametric model checking for timed automata for information flow property.
In addition, and in contrast to most of the aforementioned works, our approach comes with an implementation.

\paragraph{Execution times and timed automata}
In this paper, we need to compute execution times in timed automata, \ie{} the durations of all runs reaching the final state.
Despite its natural aspect, this problem seems sparsely investigated in the literature. Both~\cite{BDR08,Rosenmann19} deal with the computation of duration sets; we do reuse some of the reasoning from~\cite{BDR08} in our proof of \cref{proposition:ET-opacity-computation}.\label{oldtext:BDR08,Rosenmann19}
Conversely, results from~\cite{BR07} cannot be used in our work: while the equality must be forbidden in PTCTL formulae to make the problems decidable, equality constraints in a PTCTL formula would be required for such an approach to answer our problems.
Furthermore, the results presented in~\cite{BHJJM21} cannot be applied to our study, since they concern one-clock timed automata.

\paragraph{Mitigating information leakage}

Complex systems may exhibit security problems through information leakage due to the presence of unintended communication media, called side channels.
An example is time side channels in which measuring, \eg{} execution times, gives information on some sensitive information.\label{newtext:sidechannel}
In 2018, the Spectre vulnerability~\cite{KHFGGHHLMP20} exploited speculative execution to bring secret information into the cache; subsequently, cache-timing attacks were launched to exfiltrate these secrets.
Therefore, mitigation of timing attacks is of utmost importance.

Our work is related to work on mitigating information leakage through those time side channels~\cite{Aga00,MPSW05,CVBS09,WS17,WGSW18}.
	In~\cite{Aga00}, Agat \emph{et al.}\ proposed to eliminate time side channels through type-driven cross-copying.
	In~\cite{MPSW05}, Molnar~\emph{et al.}\ proposed, along the program counter model, a method for mitigating side channel through merging branches. A similar idea was proposed in~\cite{BRW06}.
	Coppens \emph{et al.}~\cite{CVBS09} developed a compiler backend for removing such leaks on x86 processors.
In~\cite{WS17}, Wang \emph{et al.}\ proposed to automatically generate masking code for eliminating side channels through program synthesis.
In~\cite{WGSW18}, Wu \emph{et al.}\ proposed to eliminate time side channels through program repair.
Different from the above-mentioned works, we reduce the problem of mitigating time side channels as a parametric model checking problem and solve it using parametric reachability analysis techniques.

This work is related to work on identifying information leakage through timing analysis~\cite{SPW18,CR11,ALKH16,ZGSW18,DSF16,DKMR15,GWW18}.
In~\cite{CR11}, Chattopadhyay and Roychoudhury applied model checking to perform cache timing analysis.
In~\cite{CJM16}, Chu \emph{et al.}\ performed similar analysis through symbolic execution.
In~\cite{ALKH16}, Abbasi \emph{et al.}\ apply the NuSMV model checker to verify integrated circuits against information leakage through side channels.
In~\cite{DKMR15}, a tool is developed to identify time side channel through static analysis. In~\cite{ZGSW18}, Sung \emph{et al.}\ developed a framework based on LLVM for cache timing analysis.
\section{Preliminaries}\label{section:preliminaries}

In this work, we assume a system is modeled in the form of a parametric timed automaton (PTA).
	In \cref{ss:Java2PTA}, we discuss how we can model programs with unknown design parameters (\eg{} a Java program with a statement \stylecode{Thread.sleep(n)} where \stylecode{n} is unknown) as PTA.

\subsection{Clocks, parameters and guards}\label{ss:clocks}

We assume a set~$\Clock = \{ \clock_1, \dots, \clock_\ClockCard \} $ of \emph{clocks}, \ie{} real-valued variables that all evolve over time at the same rate.
A clock valuation is a function
$\clockval : \Clock \rightarrow \Rgeqzero$.
We write $\ClocksZero$ for the clock valuation assigning $0$ to all clocks.
Given $d \in \Rgeqzero$, $\clockval + d$ denotes the valuation \st{} $(\clockval + d)(\clock) = \clockval(\clock) + d$, for all $\clock \in \Clock$.
Given $\resets \subseteq \Clock$, we define the \emph{reset} of a valuation~$\clockval$, denoted by $\reset{\clockval}{\resets}$, as follows: $\reset{\clockval}{\resets}(\clock) = 0$ if $\clock \in \resets$, and $\reset{\clockval}{\resets}(\clock)=\clockval(\clock)$ otherwise.

We assume a set~$\Param = \{ \param_1, \dots, \param_\ParamCard \} $ of \emph{parameters}, \ie{} unknown constants.
A parameter {\em valuation} $\pval$ is a function
$\pval : \Param \rightarrow \setqplus$.
We assume ${\compOp} \in \{<, \leq, =, \geq, >\}$.
A guard~$\guard$ is a constraint over $\Clock \cup \Param$ defined by a conjunction of inequalities of the form
$\clock \compOp \sum_{1 \leq i \leq \ParamCard} \alpha_i \param_i + d$\label{def:clockguards}, with
	$\param_i \in \Param$,
	and
	$\alpha_i, d \in \setz$.
Given~$\guard$, we write~$\clockval\models\pval(\guard)$ if %
the expression obtained by replacing each~$\clock$ with~$\clockval(\clock)$ and each~$\param$ with~$\pval(\param)$ in~$\guard$ evaluates to true.

\subsection{Parametric timed automata}

Parametric timed automata (PTAs) extend timed automata with parameters within guards and invariants in place of integer constants~\cite{AHV93}.

\subsubsection{Syntax}
\begin{definition}[PTA]\label{def:uPTA}
	A PTA $\A$ is a tuple \mbox{$\A = (\Actions, \Loc, \locinit, \locfinal, \Clock, \Param, \invariant, \Edges)$}, where:
	\begin{ienumeration}
		\item $\Actions$ is a finite set of actions,
		\item $\Loc$ is a finite set of locations,
		\item $\locinit \in \Loc$ is the initial location,
		\item $\locfinal \in \Loc$ is the (unique) final location,
		\item $\Clock$ is a finite set of clocks,
		\item $\Param$ is a finite set of parameters,
		\item $\invariant$ is the invariant, assigning to every $\loc\in \Loc$ a guard $\invariant(\loc)$,
		\item $\Edges$ is a finite set of edges  $\edge = (\loc,\guard,\action,\resets,\loc')$
		where~$\loc,\loc'\in \Loc$ are the source and target locations, $\action \in \Actions$, $\resets\subseteq \Clock$ is a set of clocks to be reset, and $\guard$ is a guard.
	\end{ienumeration}
\end{definition}
\begin{figure}[tb]

	\centering
	 \footnotesize

	\begin{tikzpicture}[scale=1, xscale=2.5, yscale=1.5, auto, ->, >=stealth']

		\node[location, initial] at (0, -.5) (s0) {$\loc_0$};

		\node[location, private] at (.75, 0) (s2) {$\loc_2$};

		\node[location, final] at (1.5, -.5) (s1) {$\loc_1$};

		\node[invariant, above=of s0] {$\styleclock{\clock} \leq 3$};
		\node[invariant, above=of s2] {$\styleclock{\clock} \leq 3$};

		\path (s0) edge[bend left] node[align=center]{$\styleclock{\clock} \geq \styleparam{\param_1}$} (s2);
		\path (s0) edge[] node[below, align=center]{$\styleclock{\clock} \geq \styleparam{\param_2}$} (s1);
		\path (s2) edge[bend left] node[align=center]{} (s1);

	\end{tikzpicture}
	\caption{A PTA example}
	\label{figure:example-PTA}

\end{figure}
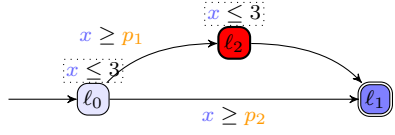
\begin{example}
	Consider the PTA in \cref{figure:example-PTA} (inspired by \cite[Fig.~1b]{GMR07}), using one clock~$\clock$ and two parameters~$\param_1$ and~$\param_2$.
	$\loc_0$ is the initial location, while we assume that~$\loc_1$ is the (only) \emph{final} location, \ie{} a location in which an attacker can measure the execution time from the initial location.
\end{example}
\paragraph{L/U-PTAs}

For some theoretical problems solved in \cref{section:theory,section:theoryFullOpacity}, we will consider the subclass of PTAs called ``lower-bound/upper-bound parametric timed automata'' (L/U-PTAs), introduced in~\cite{HRSV02}.

\begin{definition}[L/U-PTA~\cite{HRSV02}]\label{def:LUPTA}
	An \emph{L/U-PTA} is a PTA where the set of parameters is partitioned into lower-bound parameters and upper-bound parameters,
	where each upper-bound (resp.\ lower-bound) parameter~$\param_i$ must be such that,
    for every guard or invariant constraint $\clock \compOp \sum_{1 \leq i \leq \ParamCard} \alpha_i \param_i + d$, we have:
		${\compOp} \in \{ \leq, < \}$ implies $\alpha_i \geq 0$ (resp.\ $\alpha_i \leq 0$) and
		${\compOp} \in \{ \geq, > \}$ implies $\alpha_i \leq 0$ (resp.\ $\alpha_i \geq 0$).
\end{definition}
\begin{example}
	The PTA in \cref{figure:example-PTA} is an L/U-PTA with $\{ \param_1, \param_2 \}$ as lower-bound parameters, and $\emptyset$ as upper-bound parameters.

	The PTA in \cref{figure:example-Java:PTA} is not an L/U-PTA, because $\styleparam{\param}$ is compared to $\styleclock{cl}$ both as a lower-bound (in ``$\styleparam{p} \times 32^2 \leq \styleclock{cl}$'') and as an upper-bound (``$\styleclock{cl} \leq \styleparam{p} \times 32^2 + \styleparam{\epsilon}$'').
\end{example}

Given\ a parameter valuation~$\pval$, we denote by $\valuate{\A}{\pval}$ the non-parametric structure where all occurrences of any parameter~$\param_i$ have been replaced by~$\pval(\param_i)$.
	We denote as a \emph{timed automaton} any structure $\valuate{\A}{\pval}$, by assuming a rescaling of the constants: by multiplying all constants in $\valuate{\A}{\pval}$ by the least common multiple of their denominators, we obtain an equivalent (integer-valued) TA, as defined in~\cite{AD94}.

\paragraph{Synchronized product of PTAs}

The \emph{synchronous product} (using strong broadcast, \ie{} synchronization on a given set of actions), or \emph{parallel composition}, of several PTAs gives a PTA.
\begin{definition}[synchronized product of PTAs]\label{definition:parallel}
	Let $N \in \setn$.
	Given a set of PTAs $\A_i = (\Actions_i, \Loc_i, (\locinit)_i, (\locfinal)_i, \Clock_i, \Param_i, \invariant_i, \Edges_i)$, $1 \leq i \leq N$,
	and a set of actions $\Actions_s$,
	the \emph{synchronized product} of $\A_i$, $1 \leq i \leq N$,
	denoted by $\A_1 \parallel_{\Actions_s} \A_2 \parallel_{\Actions_s} \cdots \parallel_{\Actions_s} \A_N$,
	is the tuple
		$(\Actions, \Loc, \locinit, \locfinal, \Clock, \Param, \invariant, \Edges)$, where:
	\begin{enumerate}
		\item $\Actions = \bigcup_{i=1}^N\Actions_i$,
		\item $\Loc = \prod_{i=1}^N \Loc_i$,
		\item $\locinit = ((\locinit)_1, \dots, (\locinit)_N)$,
		\item $\locfinal = ((\locfinal)_1, \dots, (\locfinal)_N)$,
		\item $\Clock = \bigcup_{1 \leq i \leq N} \Clock_i$,
		\item $\Param = \bigcup_{1 \leq i \leq N} \Param_i$,
		\item $\invariant((\loc_1, \dots, \loc_N)) = \bigwedge_{i = 1}^{N} \invariant_i(\loc_i)$ for all $(\loc_1, \dots, \loc_N) \in \Loc$,
	\end{enumerate}
	and $\Edges{}$ is defined as follows.
	For all $\action \in \Actions$,
	let $\ActionsIndices_\action$ be the subset of indices $i \in 1, \dots, N$
	such that $\action \in \Actions_i$.
	For all  $\action \in \Actions$,
	for all $(\loc_1, \dots, \loc_N) \in \Loc$,
	for all \mbox{$(\loc_1', \dots, \loc_N') \in \Loc$},
	$\big((\loc_1, \dots, \loc_N), \guard, \action, \resets, (\loc'_1, \dots, \loc'_N)\big) \in \Edges$
	if:
	\begin{itemize}
		\item if $\action \in \Actions_s$, then
		\begin{enumerate}
			\item for all $i \in \ActionsIndices_\action$, there exist $\guard_i, \resets_i$ such that $(\loc_i, \guard_i, \action, \resets_i, \loc_i') \in \Edges_i$, $\guard = \bigwedge_{i \in \ActionsIndices_\action} \guard_i$, $\resets = \bigcup_{i \in \ActionsIndices_\action}\resets_i$, and,
			\item for all $i \not\in \ActionsIndices_\action$, $\loc_i' = \loc_i$.
		\end{enumerate}
		\item otherwise (if $\action \notin \Actions_s$), then there exists $i \in \ActionsIndices_\action$ such that
		\begin{enumerate}
			\item there exist $\guard_i, \resets_i$ such that $(\loc_i, \guard_i, \action, \resets_i, \loc_i') \in \Edges_i$, $\guard = \guard_i$, $\resets = \resets_i$, and,
			\item for all $j \neq i$, $\loc_j' = \loc_j$.
		\end{enumerate}
	\end{itemize}
\end{definition}

That is, synchronization is only performed on~$\Actions_s$, and other actions are interleaved.
\subsubsection{Concrete semantics of TAs}

Let us now recall the concrete semantics of TA.

\begin{definition}[Semantics of a TA]\label{def:semantics:TA}
	Given a PTA $\A = (\Actions, \Loc, \locinit, \locfinal, \Clock, \Param, \invariant, \Edges)$,
	and a parameter valuation~\(\pval\),
	the semantics of $\valuate{\A}{\pval}$ is given by the timed transition system (TTS) \cite{HMP91} $\TTS_{\valuate{\A}{\pval}}=(\States, \sinit, \flecheRel)$, with
	\begin{itemize}
		\item $\States = \{ (\loc, \clockval) \in \Loc \times \Rgeqzero^\ClockCard \mid \clockval \models \valuate{\invariant(\loc)}{\pval} \}$, %
		\item $\sinit = (\locinit, \ClocksZero) $,
		\item  $\flecheRel$ consists of the discrete and (continuous) delay transition relations:
		\begin{ienumeration}
			\item discrete transitions: $(\loc,\clockval) \longueflecheRel{\edge} (\loc',\clockval')$, %
				if $(\loc, \clockval) , (\loc',\clockval') \in \States$, and there exists ${\edge = (\loc,\guard,\action,\resets,\loc') \in \Edges}$, such that $\clockval'= \reset{\clockval}{\resets}$, and $\clockval\models\pval(\guard$).
			\item delay transitions: $(\loc,\clockval) \longueflecheRel{d} (\loc, \clockval+d)$, with $d \in \Rgeqzero$, if $\forall d' \in [0, d], (\loc, \clockval+d') \in \States$.
		\end{ienumeration}
	\end{itemize}
\end{definition}

    Moreover we write $(\loc, \clockval)\longuefleche{(d, \edge)} (\loc',\clockval')$ for a combination of a delay and discrete transition if
		$\exists  \clockval'' :  (\loc,\clockval) \longueflecheRel{d} (\loc,\clockval'') \longueflecheRel{\edge} (\loc',\clockval')$.

Given a TA~$\valuate{\A}{\pval}$ with concrete semantics $(\States, \sinit, \flecheRel)$, we refer to the states of~$\TTS_{\valuate{\A}{\pval}}$ as the \emph{concrete states} of~$\valuate{\A}{\pval}$.
A \emph{run} of~$\valuate{\A}{\pval}$ is a (finite or infinite\label{newtext:finite}) alternating sequence of concrete states of $\valuate{\A}{\pval}$ and pairs of delays and edges starting from the initial state~$\sinit$ of the form
$\concstate_0, (d_0, \edge_0), \concstate_1, \cdots$
with
$i = 0, 1, \dots$, $\edge_i \in \Edges$, $d_i \in \Rgeqzero$ and
	$\concstate_i \longuefleche{(d_i, \edge_i)} \concstate_{i+1}$.

Given a state~$\concstate = (\loc, \clockval)$, we say that $\concstate$ is reachable in~$\valuate{\A}{\pval}$ if $\concstate$ appears in a run of $\valuate{\A}{\pval}$.
By extension, we say that $\loc$ is reachable; and by extension again, given a set~$\somelocs$ of locations, we say that $\somelocs$ is reachable if there exists $\loc \in \somelocs$ such that $\loc$ is reachable in~$\valuate{\A}{\pval}$.\footnote{%
	We use an existential quantification over the set~$\somelocs$ of locations, \ie{} the set of locations is reachable if at least one target location is reachable.
	This is a standard definition for reachability synthesis (see, \eg{} \cite{JLR15}), and our \cref{algo:SynthOp} uses a singleton set for $\somelocs$ anyway.
	\label{newtext:EF:existential}
}
Given $\loc, \loc' \in \Loc$ and a run~$\varrun$, we say that $\loc$ is reached on the way to~$\loc'$ in~$\varrun$ if $\varrun$ is of the form $(\loc_0, \clockval_0), (d_0, \edge_0), (\loc_1, \clockval_1), \cdots, (\loc_m, \clockval_m), (d_m, \edge_m), \cdots (\loc_n, \clockval_n)$ %
	for some~$m,n \in \setn$ such that $\loc_m = \loc$, $\loc_n = \loc'$ and $\forall 0 \leq i \leq m-1, \loc_i \neq \loc'$.
Conversely, $\loc$ is avoided on the way to~$\loc'$ in~$\varrun$ if $\varrun$ is of the form $(\loc_0, \clockval_0), (d_0, \edge_0), (\loc_1, \clockval_1), \cdots, (\loc_n, \clockval_n )$ %
	with $\loc_n = \loc'$ and $\forall 0 \leq i \leq n, \loc_i \neq \loc$.
Given $\loc, \loc' \in \Loc$, we say that $\loc$ is \emph{reachable on the way to~$\loc'$} in~$\valuate{\A}{\pval}$ if there exists a run~$\varrun$ of~$\valuate{\A}{\pval}$ for which $\loc$ is reached on the way to~$\loc'$ in~$\varrun$.
Otherwise, $\loc$ is \emph{unreachable on the way to~$\loc'$}.\label{text:reachableontheway}

The \emph{duration} of a finite run $\varrun : \concstate_0, (d_0, \edge_0), \concstate_1, \cdots, (d_{i-1}, \edge_{i-1}), (\loc_i, \clockval_i)$ is $\duration(\varrun) = \sum_{0 \leq j \leq i-1} d_j$.
We also say that $\loc_i$ is reachable in time~$\duration(\varrun)$.
\begin{example}
	Consider again the PTA~$\A$ in \cref{figure:example-PTA}, and let $\pval$ be such that $\pval(\param_1) = 1$ and $\pval(\param_2) = 2$.
	Consider the following run~$\varrun$ of $\valuate{\A}{\pval}$:
	$(\loc_0, \clock = 0) , (1.4, \edge_2) , (\loc_2, \clock = 1.4) , (1.3, \edge_3) , (\loc_1, \clock = 2.7)$, where
		$\edge_2$ is the edge from $\loc_0$ to~$\loc_2$ in \cref{figure:example-PTA},
		and
		$\edge_3$ is the edge from $\loc_2$ to~$\loc_1$.
		We write ``$\clock = 1.4$'' instead of ``$\clockval$ such that $\clockval(\clock) = 1.4$''.
	We have $\duration(\varrun) = 1.4 + 1.3 = 2.7$.
	In addition, $\loc_2$ is reached on the way to~$\loc_1$ in~$\varrun$.
\end{example}
\subsubsection{Timed automata regions}\label{sss:regions}

Let us next recall the concept of regions and the region graph~\cite{AD94}.

Given a TA~$\valuate{\A}{\pval}$, for a clock $\clock_i$, we denote by $c_i$ the largest constant to which $\clock_i$ is compared within the guards and invariants of~$\valuate{\A}{\pval}$ (that is, $c_i = \max_{i}( \{\ d_i \mid \clock \compOp d_i \text{~appears in a guard or invariant of~}\valuate{\A}{\pval} \}$).
Given a clock valuation $\clockval$ and a clock~$\clock_i$, let $\integralp{\clockval(\clock_i)}$ and $\fract{\clockval(\clock_i)}$ denote respectively the integral part and the fractional part of~$\clockval(\clock_i)$.
\begin{example}\label{example:largestConstant}
	Consider again the PTA in \cref{figure:example-PTA}, and let $\pval$ be such that $\pval(\param_1) = 2$ and $\pval(\param_2) = 4$.
	In the TA $\valuate{\A}{\pval}$, the clock $\clock$ is compared to the constants in $\set{2, 3, 4}$. 
	In that case, $c = 4$ is the largest constant to which the clock $\clock$ is compared.
\end{example}
\begin{definition}[Region equivalence]\label{definition:region-equivalence}
	We say that two clock valuations $\clockval$ and $\clockval'$ are equivalent,
	denoted $\clockval \equivalent \clockval'$,
	if the following three conditions hold for any clocks $x_i,x_j$:
	\begin{enumerate}
		\item either
		\begin{ienumeration} 
			\item $\integralp{\clockval(\clock_i)} = \integralp{\clockval'(\clock_i)}$
			or 
			\item $\clockval(\clock_i)>c_i$ and $\clockval'(\clock_i)>c_i$
		\end{ienumeration}
		\item $\fract{\clockval(\clock_i)}\leq\fract{\clockval(\clock_j)}$ iff $\fract{\clockval'(\clock_i)}\leq\fract{\clockval'(\clock_j)}$
		\item $\fract{\clockval(\clock_i)}=0$ iff $\fract{\clockval'(\clock_i)}=0$
	\end{enumerate}
\end{definition}
The equivalence relation $\equivalent$ is extended to the states of $\TTS_{\valuate{\A}{\pval}}$: if ${\concstate = (\loc,\clockval)}, {\concstate' = (\loc',\clockval')}$ are two states of $\TTS_{\valuate{\A}{\pval}}$, we write ${\concstate\equivalent \concstate'}$ iff ${\loc=\loc'}$ and ${\clockval\equivalent \clockval'}$.

We denote by $\class{\concstate}$ the equivalence class of $\concstate$ for $\equivalent$. %
A \emph{region} is an equivalence class $\class{\concstate}$ of $\equivalent$.
The set of all regions is denoted $\Regions_{\valuate{\A}{\pval}}$.
Given a state $\concstate = (\loc, \clockval)$ and $d \geq 0$, we write $\concstate + d$ to denote $(\loc, \clockval + d)$.%

\begin{definition}[Region graph \cite{BDR08}]\label{definition:region-graph}
	The \emph{region graph} ${\RegionGraph_{\valuate{\A}{\pval}} = (\Regions_{\valuate{\A}{\pval}},\regionEdges_{\valuate{\A}{\pval}})}$ is a finite graph with:
	\begin{itemize}
		\item $\Regions_{\valuate{\A}{\pval}}$ as the set of vertices
		\item given two regions $\region=\class{\concstate},\region'=\class{\concstate'}\in\Regions_{\valuate{\A}{\pval}}$, we have $(\region,\region')\in\regionEdges_{\valuate{\A}{\pval}}$ if one of the following holds:
		\begin{itemize}
			\item if $\concstate \longueflecheRel{\edge} \concstate' \in \TTS_{\valuate{\A}{\pval}}$ for some $\edge \in \Edges$ (\emph{discrete instantaneous transition});
			
			\item if $\region'$ is a time successor of $\region$: $r\neq r'$ and there exists $d$ such that $\concstate + d\in \region'$ and ${\forall d'<d, \concstate+d'\in \region\cup\region'}$  (\emph{delay transition});

			\item $\region=\region'$ is unbounded: $\concstate = (\loc,\clockval)$ with $\clockval(\clock_i)>c_i$ for all~$x_i$  (\emph{equivalent unbounded regions}).
		\end{itemize}
	\end{itemize}
\end{definition}

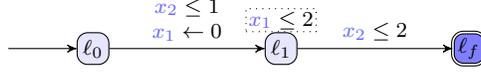
\begin{figure}
	\centering
	\footnotesize

	\begin{tikzpicture}[scale=1, xscale=2.5, yscale=2.2, auto, ->, >=stealth']
	
	\node[location, initial] at (0, 0) (l0) {$\locinit$};
	
	\node[location] at (1, 0) (l1) {$\loc_1$};
	
	\node[location, final] at (2, 0) (l2) {$\locfinal$};

	\node[invariant, above=of l1] {$\styleclock{\clock_1} \leq 2$};
	
	\path (l0) edge node[align=center]{$\styleclock{\clock_2} \leq 1$\\$ \styleclock{\clock_1} \assign 0$} (l1);
	\path (l1) edge node[align=center]{$\styleclock{\clock_2} \leq 2$} (l2);
	
	\end{tikzpicture}
	\caption{A TA example}
	\label{figure:example-TA}
\end{figure}

\begin{figure}
	\begin{subfigure}{.3\textwidth}
		\begin{adjustbox}{width=\textwidth}
			\begin{tikzpicture}[line cap=round,line join=round,>=triangle 45,x=1cm,y=1cm]
			\begin{axis}[
			x=1cm,y=1cm,
			axis lines=middle,
			ymajorgrids=true,
			xmajorgrids=true,
			xmin=0,
			xmax=4,
			ymin=0,
			ymax=4,
			xtick={0,1,...,4},
			ytick={0,1,...,4},
			xlabel=$\clock_1$,
			ylabel=$\clock_2$,
			]
			
			\clip(0,0) rectangle (4,4);
			
			\draw [line width=2pt, color=blue] (0,0)-- (1,1);
			\end{axis}
			\end{tikzpicture}
		\end{adjustbox}
		\caption{Step 1: $\clock_1$ and $\clock_2$ linearly increase as long as $y\leq 1$}
	\label{figure:example-region-graph:1}
	\end{subfigure}
	\ 
	\begin{subfigure}{.3\textwidth}
		\begin{adjustbox}{width=\textwidth}
			\begin{tikzpicture}[line cap=round,line join=round,>=triangle 45,x=1cm,y=1cm]
			\begin{axis}[
			x=1cm,y=1cm,
			axis lines=middle,
			ymajorgrids=true,
			xmajorgrids=true,
			xmin=0,
			xmax=4,
			ymin=0,
			ymax=4,
			xtick={0,1,...,4},
			ytick={0,1,...,4},
			xlabel=$\clock_1$,
			ylabel=$\clock_2$,
			]
			
			\clip(0,0) rectangle (4,4);
			
			\draw [line width=2pt, color=red] (0,0)-- (0,1);
			\end{axis}
			\end{tikzpicture}
		\end{adjustbox}
		\caption{Step 2: $\clock_1$ is reset}
	\label{figure:example-region-graph:2}
	\end{subfigure}
	\ 
	\begin{subfigure}{.3\textwidth}
		\begin{adjustbox}{width=\textwidth}
			\begin{tikzpicture}[line cap=round,line join=round,>=triangle 45,x=1cm,y=1cm]
			\begin{axis}[
			x=1cm,y=1cm,
			axis lines=middle,
			ymajorgrids=true,
			xmajorgrids=true,
			xmin=0,
			xmax=4,
			ymin=0,
			ymax=4,
			xtick={0,1,...,4},
			ytick={0,1,...,4},
			xlabel=$\clock_1$,
			ylabel=$\clock_2$,
			]
			
			\clip(0,0) rectangle (4,4);
			
			\fill[line width=2pt,color=green,fill=green,pattern=north west lines,pattern color=green] (0,0) -- (0,1) -- (2,3) -- (2,2) -- cycle;
			\end{axis}
			\end{tikzpicture}
		\end{adjustbox}
		\caption{Step 3: $\clock_1$ and $\clock_2$ linearly increase as long as $\clock_1\leq 2$}
	\end{subfigure}
	
	\begin{subfigure}{.3\textwidth}
		\begin{adjustbox}{width=\textwidth}
			\begin{tikzpicture}[line cap=round,line join=round,>=triangle 45,x=1cm,y=1cm]
			\begin{axis}[
			x=1cm,y=1cm,
			axis lines=middle,
			ymajorgrids=true,
			xmajorgrids=true,
			xmin=0,
			xmax=4,
			ymin=0,
			ymax=4,
			xtick={0,1,...,4},
			ytick={0,1,...,4},
			xlabel=$\clock_1$,
			ylabel=$\clock_2$,
			]
			
			\clip(0,0) rectangle (4,4);
			
			\fill[line width=2pt,color=orange,fill=orange,pattern=north west lines,pattern color=orange] (0,0) -- (0,1) -- (1,2) -- (2,2) -- cycle;
			\end{axis}
			\end{tikzpicture}
		\end{adjustbox}
		\caption{Step 4: $\clock_2\leq 2$}
	\end{subfigure}
	\ 
	\begin{subfigure}{.3\textwidth}
		\begin{adjustbox}{width=\textwidth}
			\begin{tikzpicture}[line cap=round,line join=round,>=triangle 45,x=1cm,y=1cm]
			\begin{axis}[
			x=1cm,y=1cm,
			axis lines=middle,
			ymajorgrids=true,
			xmajorgrids=true,
			xmin=0,
			xmax=4,
			ymin=0,
			ymax=4,
			xtick={0,1,...,4},
			ytick={0,1,...,4},
			xlabel=$\clock_1$,
			ylabel=$\clock_2$,
			]
			
			\clip(0,0) rectangle (4,4);
			
			\fill[line width=2pt,color=orange,fill=orange,pattern=north west lines,pattern color=orange] (0,0) -- (0,1) -- (1,1) -- cycle;
			\draw [line width=2pt, color=green] (0,1)-- (1,1);
			\fill[line width=2pt,color=blue,fill=blue,pattern=north east lines,pattern color=blue] (0,1) -- (1,1) -- (1,2) -- cycle;
			\draw [line width=2pt, color=purple] (1,1) -- (1,2);
			\fill[line width=2pt,color=red,fill=red,pattern=north west lines,pattern color=red] (1,1) -- (1,2) -- (2,2) -- cycle;
			\end{axis}
			\end{tikzpicture}
		\end{adjustbox}
		\caption{Division into regions of the last constraint}
	\end{subfigure}
	\caption{Region abstraction of the TA in~\cref{figure:example-TA}}
	\label{figure:example-region-graph}
\end{figure}
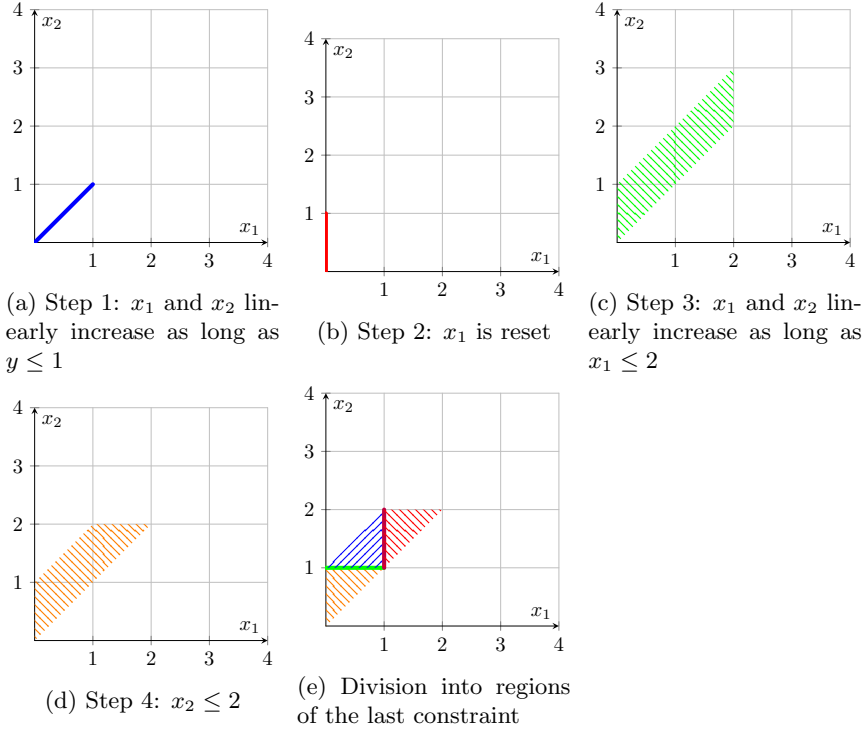

\begin{example}[Region abstraction]
	\label{example:region-abstraction}
	Consider the TA in~\cref{figure:example-TA}.
	We present in~\cref{figure:example-region-graph} some of its regions.
	For example, the region depicted in \cref{figure:example-region-graph:2} is the successor (by taking the discrete transition from~$\loc_0$ to~$\loc_1$, resetting $\clock_1$) of the region depicted in \cref{figure:example-region-graph:1}.

	Consider the valuations $\clockval_1$ and~$\clockval_2$ such that $\clockval_1(\clock_1) = 0$, $\clockval_1(\clock_2) = 0.35$, $\clockval_2(\clock_1) = 0$, and $\clockval_2(\clock_2) = 0.75$.
	These two valuations satisfy all three conditions of \cref{definition:region-equivalence}, and therefore $\clockval_1 \equivalent \clockval_2$.
	Note that they both belong to the region depicted in \cref{figure:example-region-graph:2}, and any pair of valuations in this region is (by definition) equivalent.
\end{example}

\subsection{Symbolic semantics}\label{ss:symbolic}

Let us now recall the symbolic semantics of PTAs (see \eg{} \cite{HRSV02,ACEF09}).

\paragraph{Constraints}
We first need to define operations on constraints.
A linear term over $\Clock \cup \Param$ is of the form $\sum_{1 \leq i \leq \ClockCard} \alpha_i \clock_i + \sum_{1 \leq j \leq \ParamCard} \beta_j \param_j + d$, with
	$\clock_i \in \Clock$,
	$\param_j \in \Param$,
	and
	$\alpha_i, \beta_j, d \in \setz$.
A \emph{constraint}~$\Constr$ (\ie{} a convex polyhedron\footnote{%
	Strictly speaking, we manipulate \emph{polytopes}, while polyhedra refer to 3-dimensional polytopes.
	However, for sake of consistency with the parametric timed model checking literature, and with the Parma polyhedra library (among others), we refer to these geometric objects as \emph{polyhedra}.
}) over $\Clock \cup \Param$ is a conjunction of inequalities of the form $\lterm \compOp 0$, where $\lterm$ is a linear term.

Given a parameter valuation~$\pval$, $\valuate{\Constr}{\pval}$ denotes the constraint over~$\Clock$ obtained by replacing each parameter~$\param$ in~$\Constr$ with~$\pval(\param)$.
Likewise, given a clock valuation~$\clockval$, $\valuate{\valuate{\Constr}{\pval}}{\clockval}$ denotes the expression obtained by replacing each clock~$\clock$ in~$\valuate{\Constr}{\pval}$ with~$\clockval(\clock)$.
We write $\clockval \models \valuate{\Constr}{\pval}$ whenever $\valuate{\valuate{\Constr}{\pval}}{\clockval}$ evaluates to true.
We say that %
$\pval$ \emph{satisfies}~$\Constr$,
denoted by $\pval \models \Constr$,
if the set of clock valuations satisfying~$\valuate{\Constr}{\pval}$ is nonempty.
We say that $\Constr$ is \emph{satisfiable} if $\exists \clockval, \pval \text{ s.t.\ } \clockval \models \valuate{\Constr}{\pval}$.

We define the \emph{time elapsing} of~$\Constr$, denoted by $\timelapse{\Constr}$, as the constraint over $\Clock$ and $\Param$ obtained from~$\Constr$ by delaying all clocks by an arbitrary amount of time.
That is,
\[\clockval' \models \valuate{\timelapse{\Constr}}{\pval} \text{ if } \exists \clockval : \Clock \to \setrplus, \exists d \in \setrplus \text { s.t. } \clockval \models \valuate{\Constr}{\pval} \land \clockval' = \clockval + d \text{.}\]
Given $\resets \subseteq \Clock$, we define the \emph{reset} of~$\Constr$, denoted by $\reset{\Constr}{\resets}$, as the constraint obtained from~$\Constr$ by resetting the clocks in~$\resets$ to $0$, and keeping the other clocks unchanged.
That is,
\[\clockval' \models \valuate{\reset{\Constr}{\resets}}{\pval} \text{ if } \exists \clockval : \Clock \to \setrplus \text { s.t. } \clockval \models \valuate{\Constr}{\pval} \land \forall \clock \in \Clock
	\left \{ \begin{array}{ll}
		 \clockval'(\clock) = 0 & \text{if } \clock \in \resets\\
		 \clockval'(\clock) = \clockval(\clock) & \text{otherwise.}
	\end{array} \right .\]
We denote by $\projectP{\Constr}$ the projection of~$\Constr$ onto~$\Param$, \ie{} obtained by eliminating the variables not in~$\Param$ (\eg{} using Fourier-Motzkin~\cite{Schrijver86}).

\begin{definition}[Symbolic state]
	A symbolic state is a pair $(\loc, \Constr)$ where $\loc \in \Loc$ is a location, and $\Constr$ its associated parametric zone.
\end{definition}
\begin{definition}[Symbolic semantics]\label{def:PTA:symbolic}
	Given a PTA $\A = (\Actions, \Loc, \locinit, \locfinal,
	\Clock, \Param, \invariant, \Edges)$,
	the symbolic semantics of~$\A$ is the labeled transition system called \emph{parametric zone graph}
	$ \PZG = ( \Edges, \SymbState, \symbstateinit, \symbtrans )$, with
	\begin{itemize}
		\item $\SymbState = \{ (\loc, \Constr) \mid \Constr \subseteq \invariant(\loc) \}$, %
		\item $\symbstateinit = \big(\locinit, \timelapse{(\bigwedge_{1 \leq i\leq\ClockCard}\clock_i=0)} \land \invariant(\loc_0) \big)$,
				and
		\item $\big((\loc, \Constr), \edge, (\loc', \Constr')\big) \in \symbtrans $ if $\edge = (\loc,\guard,\action,\resets,\loc') \in \Edges$ and
			\[\Constr' = \timelapse{\big(\reset{(\Constr \land \guard)}{\resets}\land \invariant(\loc')\big )} \land \invariant(\loc')\]
			with $\Constr'$ satisfiable.
	\end{itemize}

\end{definition}

That is, in the parametric zone graph, nodes are symbolic states, and arcs are labeled by \emph{edges} of the original PTA.

If $(\symbstate, \edge, \symbstate' ) \in \symbtrans $, we write $\Succ(\symbstate, \edge) = \symbstate'$.
By extension, we write $\Succ(\symbstate)$ for $\cup_{\edge \in \Edges} \Succ(\symbstate, \edge)$.\label{newtext:succ}

In the non-parametric timed automata setting, a zone can be seen as a \emph{convex union} of regions.
In the parametric setting, a parametric zone constrains both clocks and parameters in such a way that, for each admissible parameter valuation, the resulting projection on clocks is a zone.
\label{newtext:zones}

\begin{example}
	Consider again the PTA~$\A$ in \cref{figure:example-PTA}.
	The parametric zone graph of~$\A$ is given in \cref{figure:example-PTA:PZG}, where
		$\edge_1$ is the edge from $\loc_0$ to~$\loc_1$ in \cref{figure:example-PTA},
		$\edge_2$ is the edge from $\loc_0$ to~$\loc_2$,
		and
		$\edge_3$ is the edge from $\loc_2$ to~$\loc_1$.
	In addition, the symbolic states are:

	\begin{tabular}{r @{ } l  @{ }c  @{ }l  @{ }l}
		$\symbstate_0 =($ & $\loc_0$ & $,$ & $0 \leq \clock \leq 3 \land \param_1 \geq 0 \land \param_2 \geq 0 $ & $)$\\
		$\symbstate_1 =($ & $\loc_1$ & $,$ & $\clock \geq \param_2 \land 0 \leq \param_2 \leq 3 \land \param_1 \geq 0 $ & $)$\\
		$\symbstate_2 =($ & $\loc_2$ & $,$ & $3 \geq \clock \geq \param_1 \land 0 \leq \param_1 \leq 3 \land \param_2 \geq 0 $ & $)$\\
		$\symbstate_3 =($ & $\loc_1$ & $,$ & $\clock \geq \param_1 \land 0 \leq \param_1 \leq 3 \land \param_2 \geq 0 $ & $)$
		.
	\end{tabular}
\end{example}
\begin{figure}[tb]

	\centering
	 \footnotesize

	\begin{tikzpicture}[scale=1, xscale=2.5, yscale=2.2, auto, ->, >=stealth']

		\node[location, initial] at (0, 0) (s0) {$\symbstate_0$};

		\node[location] at (1, 0) (s2) {$\symbstate_2$};

		\node[location] at (0, -.5) (s1) {$\symbstate_1$};

		\node[location] at (1, -.5) (s3) {$\symbstate_3$};

		\path (s0) edge node[align=center]{$\edge_1$} (s1);
		\path (s0) edge[] node[align=center]{$\edge_2$} (s2);
		\path (s2) edge[] node[align=center]{$\edge_3$} (s3);

	\end{tikzpicture}
	\caption{Parametric zone graph of \cref{figure:example-PTA}}
	\label{figure:example-PTA:PZG}

\end{figure}
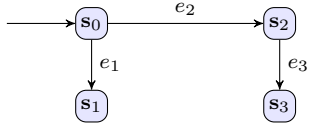
\subsection{Reachability synthesis}

We use reachability synthesis to solve the problems defined in \cref{section:problem}.
This procedure, called \EFsynth{}, takes as input a PTA~$\A$ and a set of target locations~$\somelocs$, and attempts to synthesize all parameter valuations~$\pval$ for which~$\somelocs$ is reachable in~$\valuate{\A}{\pval}$.
$\EFsynth(\A, \somelocs)$ was formalized in \eg{} \cite{JLR15} and is a procedure that may not terminate, but that computes an exact result (sound and complete) if it terminates.
\EFsynth{} traverses the \emph{parametric zone graph} of~$\A$.
\begin{example}
	Consider again the PTA~$\A$ in \cref{figure:example-PTA}.
	$\EFsynth(\A, \{ \loc_1 \}) = \param_1 \leq 3 \lor \param_2 \leq 3$.
	Intuitively, it corresponds to all parameter constraints in the parametric zone graph in \cref{figure:example-PTA:PZG} associated to symbolic states with location~$\loc_1$.
\end{example}

We finally recall the correctness of \EFsynth{}.

\begin{lemma}[\cite{JLR15}]\label{prop:EFsynth}
	Let $\A$ be a PTA, and let $\somelocs$ be a subset of the locations of~$\A$.
	Assume $\EFsynth(\A, \somelocs)$ terminates with result~$\K$.
	Then $\pval \models \K$ iff $\somelocs$ is reachable in~$\valuate{\A}{\pval}$.
\end{lemma}
\section{Timed opacity problems}\label{section:problem}
\subsection{Definitions}

Let us first introduce two key concepts to define our notion of opacity.
$\PrivDurReach{\valuate{\A}{\pval}}{\loc}{\loc'}$ (resp.\ $\PubDurReach{\valuate{\A}{\pval}}{\loc}{\loc'}$) is the set of all the durations of the runs for which $\loc$ is reachable (resp.\ unreachable) on the way to~$\loc'$.
Formally:
\[\PrivDurReach{\valuate{\A}{\pval}}{\loc}{\loc'} = \{ d \mid \exists \varrun \text{ in } \valuate{\A}{\pval} \text{ such that } d = \duration(\varrun) \land \loc \text{ is reached on the way to~$\loc'$ in~}\varrun \}\]
and
\[\PubDurReach{\valuate{\A}{\pval}}{\loc}{\loc'} = \{ d \mid \exists \varrun \text{ in } \valuate{\A}{\pval} \text{ such that } d = \duration(\varrun) \land \loc \text{ is avoided on the way to~$\loc'$ in~}\varrun \}\text{.}\]

These concepts can be seen as the set of execution times from the initial location~$\locinit$ to a target location $\loc'$ while passing (resp.\ not passing) by a private location~$\loc$.
Observe that, from the definition of $\duration(\varrun)$, this ``execution time'' does not include the time spent in~$\loc'$.

\begin{example}\label{example:running:DurReach}
	Consider again the PTA in \cref{figure:example-PTA}, and let $\pval$ be such that $\pval(\param_1) = 1$ and $\pval(\param_2) = 2$.
	We have $\PrivDurReach{\valuate{\A}{\pval}}{\loc_2}{\loc_1} = [1, 3]$
	and
	$\PubDurReach{\valuate{\A}{\pval}}{\loc_2}{\loc_1} = [2, 3]$.
\end{example}

We now introduce the concept of ``timed opacity \wrt{} a set of durations (or execution times) $\Times$'': a system is opaque \wrt{} a given location $\locpriv$ on the way to~$\locfinal$ for execution times~$\Times$ whenever, for any duration in~$\Times$, it is not possible to deduce whether the system went through~$\locpriv$ or not.
In other words, if an attacker measures an execution time within~$\Times$ from the initial location to the target location~$\locfinal$, then this attacker is not able to deduce whether the system visited~$\locpriv$.

\begin{definition}[timed opacity \wrt{} $\Times$]\label{definition:ET-timed opacity}
	Given a TA~$\valuate{\A}{\pval}$, a private location~$\locpriv$,
		a target location~$\locfinal$
		and a set of execution times~$\Times$,
		we say that $\valuate{\A}{\pval}$ is \emph{opaque \wrt{} $\locpriv$ on the way to~$\locfinal$ for execution times~$\Times$}
		if $\Times \subseteq \PrivDurReach{\valuate{\A}{\pval}}{\locpriv}{\locfinal} \cap \PubDurReach{\valuate{\A}{\pval}}{\locpriv}{\locfinal}$.
\end{definition}
If one does not have the ability to tune the system (\ie{} change internal delays, or add some \stylecode{sleep()} or \stylecode{Wait()} statements in the program), one may be first interested in knowing whether the system is opaque for all execution times (\ie{} the durations of the runs from the initial location to the target location $\locfinal$).
In other words, if a system is \emph{fully opaque}, for any possible measured execution time, an attacker is not able to deduce anything on the system, in terms of visit of~$\locpriv$.

\begin{definition}[full timed opacity]\label{definition:opacity}
	Given a TA~$\valuate{\A}{\pval}$, a private location~$\locpriv$ and
	a target location~$\locfinal$,
	we say that $\valuate{\A}{\pval}$ is \emph{fully opaque \wrt{} $\locpriv$ on the way to~$\locfinal$}
	if $\PrivDurReach{\valuate{\A}{\pval}}{\locpriv}{\locfinal} = \PubDurReach{\valuate{\A}{\pval}}{\locpriv}{\locfinal}$.
\end{definition}

That is, a system is fully opaque if, for any execution time~$d$, a run of duration~$d$ reaches~$\locfinal$ after going through~$\locpriv$ iff another run of duration~$d$ reaches~$\locfinal$ without going through~$\locpriv$.

\begin{remark}\label{remark:R2}
	This definition is symmetric: a system is not opaque iff an attacker can deduce $\locpriv$ or $\neg \locpriv$.
	For instance, if there is no path through $\locpriv$ to $\locfinal$, but a path to $\locfinal$, a system is not opaque \wrt{} \cref{definition:opacity}.
\end{remark}
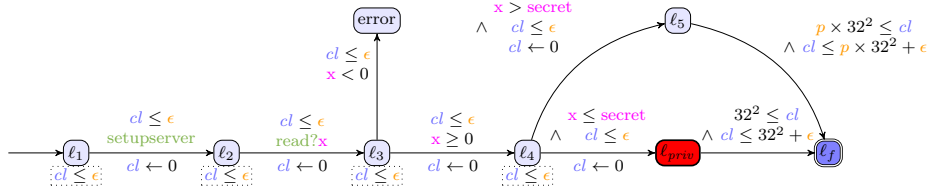
\begin{figure*}[tb]
\newcommand{\ratio}{0.5\textwidth}
 
	\centering
	 \footnotesize

\scalebox{.8}{

	\begin{tikzpicture}[scale=1, xscale=2.5, yscale=2.2, auto, ->, >=stealth']
 
		\node[location, initial] at (0,0) (l1) {$\loc_1$};
 
		\node[location] at (1, 0) (l2) {$\loc_2$};
 
		\node[location] at (2, 0) (l3) {$\loc_3$};
 
		\node[location] at (2, 1) (final1) {\styleloc{error}};
 
		\node[location] at (3, 0) (l4) {$\loc_4$};
 
		\node[location, private] at (4, 0) (hidden1) {$\locpriv$};
 
		\node[location] at (4, 1) (hidden2) {$\loc_5$};
 
		\node[location, final] at (5, 0) (final2) {$\locfinal$};
 
		\node[invariant, below=of l1] {$\clockcl \leq \styleparam{\epsilon}$};
		\node[invariant, below=of l2] {$\clockcl \leq \styleparam{\epsilon}$};
		\node[invariant, below=of l3] {$\clockcl \leq \styleparam{\epsilon}$};
		\node[invariant, below=of l4] {$\clockcl \leq \styleparam{\epsilon}$};
		
		\path (l1) edge node[align=center]{$\clockcl \leq \styleparam{\epsilon}$ \\ $\styleact{\mathrm{setupserver}}$} node[below] {$\clockcl \assign 0$} (l2);
 
		\path (l2) edge node[align=center]{$\clockcl \leq \styleparam{\epsilon}$ \\  $\styleact{\mathrm{read}?}\styledisc{x}$} node[below] {$\clockcl \assign 0$} (l3);

		\path (l3) edge node[above left, align=center]{$\clockcl \leq \styleparam{\epsilon}$ \\  $ \styledisc{x} < 0$} (final1);

		\path (l3) edge node[align=center]{$\clockcl \leq \styleparam{\epsilon}$ \\  $ \styledisc{x} \geq 0$ } node[below] {$\clockcl \assign 0$} (l4);

		\path (l4) edge node{\begin{tabular}{@{} c @{\ } c@{} }
		& $ \styledisc{x} \leq \mathrm{\styledisc{secret}}$\\
		 $\land $ & $ \clockcl \leq \styleparam{\epsilon}$\\
		\end{tabular}} node [below, align=center]{$\clockcl\assign 0$} (hidden1);

		\path (l4) edge[bend left] node{\begin{tabular}{@{} c @{\ } c@{} }
		& $ \styledisc{x} > \styledisc{\mathrm{secret}}$\\
		 $\land $ & $ \clockcl \leq \styleparam{\epsilon}$\\
		 & $\clockcl\assign 0$\\
		\end{tabular}} (hidden2);

		\path (hidden1) edge node{\begin{tabular}{@{} c @{\ } c@{} }
		& $ 32^2  \leq \clockcl$\\
		$\land$ & $ \clockcl \leq 32^2 + \styleparam{\epsilon}$\\
		\end{tabular}} (final2);

		\path (hidden2) edge[bend left] node{\begin{tabular}{@{} c @{\ } c@{} }
		& $ \styleparam{p} \times 32^2  \leq \clockcl$\\
		$\land$ & $ \clockcl \leq \styleparam{p} \times 32^2 + \styleparam{\epsilon}$\\
		\end{tabular}} (final2);
 
	\end{tikzpicture}
	
	}
	
	\caption{A Java program encoded in a PTA}
	\label{figure:example-Java:PTA}

\end{figure*}

\begin{example}\label{example:Java-PTA}
	Consider the PTA~$\A$ in \cref{figure:example-Java:PTA} where $\styleclock{cl}$ is a clock, while $\styleparam{\epsilon},\styleparam{p}$ are parameters. %
	We use a sightly extended PTA syntax:
	$\styleact{read?}\styledisc{x}$ reads the value input on a given channel~$\styleact{read}$, and assigns it to a (discrete, global) variable~$\styledisc{x}$.
	$\styledisc{secret}$ is a constant variable of arbitrary value.
	If both $\styledisc{x}$ and $\styledisc{secret}$ are finite-domain variables (\eg{} bounded integers) then they can be seen as syntactic sugar for locations.
	Such variables are supported by most model checkers, including \uppaal{}~\cite{LPY97} and \imitator{}~\cite{Andre21}.

	This PTA encodes a server process
	and is a (manual) translation of a Java program from the DARPA Space/Time Analysis for Cybersecurity (STAC) library\footnote{%
		\url{https://github.com/Apogee-Research/STAC/blob/master/Canonical_Examples/Source/Category1_vulnerable.java}
	},
	that compares a user-input variable with a given secret
	and performs different actions taking different times depending on this secret.
	The original Java program is \emph{vulnerable}, and tagged as such in the DARPA library, because some sensitive information can be deduced from the timing information.
	The original Java code is given in \cref{appendix:Java}.
	
	In our encoding, a single instruction takes a time in $[0, \styleparam{\epsilon}]$, while $\styleparam{p}$ is a (parametric) factor to one of the \stylecode{sleep} instructions of the program.
	Note that in the original Java code in \cref{appendix:Java}, at line~25, there is no parameter~$\styleparam{p}$ but an integer~2; that is, the code is fixed to have $\pval(\styleparam{p}) = 2$.\label{newtext:p=2}
	For sake of simplicity, we abstract away instructions not related to time, and merge subfunctions calls.
	For this work, we simplify the problem and abstract in this way.
	Precisely modeling the timing behavior of a program is itself a complicated problem (due to caching, speculative execution, etc.)\ and we leave to future work.

	Let $\pval_1$ such that $\pval_1(\styleparam{\epsilon}) = 1$ and $\pval_1(\styleparam{p}) = 2$%
	.
	For this example,
		$\PrivDurReach{\valuate{\A}{\pval_1}}{\locpriv}{\locfinal} = [1024, 1029]$
		while
		$\PubDurReach{\valuate{\A}{\pval_1}}{\locpriv}{\locfinal} = [2048, 2053]$.
	Therefore, $\valuate{\A}{\pval_1}$ is opaque \wrt{} $\locpriv$ on the way to~$\locfinal$ for execution times~$\Times = [1024, 1029] \cap [2048, 2053] = \emptyset$.

	Let $\pval_2$ such that $\pval_2(\styleparam{\epsilon}) = 2$ and $\pval_2(\styleparam{p}) = 1.002$.
		$\PrivDurReach{\valuate{\A}{\pval_2}}{\locpriv}{\locfinal} = [1024, 1034]$
		while
		$\PubDurReach{\valuate{\A}{\pval_2}}{\locpriv}{\locfinal} = [1026.048, 1036.048]$.
	Therefore, $\valuate{\A}{\pval_2}$ is opaque \wrt{} $\locpriv$ on the way to~$\locfinal$ for execution times~$\Times = [1026.048, 1034]$.

	Obviously, neither $\valuate{\A}{\pval_1}$ nor $\valuate{\A}{\pval_2}$ are fully opaque \wrt{} $\locpriv$ on the way to~$\locfinal$.

\end{example}
\subsection{Decision and computation problems}\label{ss:problems}

We can now define the timed opacity computation problem, which consists in computing the possible execution times ensuring opacity \wrt{} a private location.
In other words, the attacker model is as follows: the attacker knows the system model in the form of a TA, and can only observe the computation time between the start of the program and the time it reaches a given (final) location.

\smallskip

\defProblem
	{Timed opacity computation}
	{A TA~$\valuate{\A}{\pval}$, a private location~$\locpriv$,
		a target location~$\locfinal$
		}
	{Compute the execution times~$\Times$ such that $\valuate{\A}{\pval}$ is opaque \wrt{} $\locpriv$ on the way to~$\locfinal$ for these execution times~$\Times$.}

\medskip

Let us illustrate that this computation problem is certainly not easy.
For the TA~$\A$ in \cref{figure:opaqueforN}, the execution times~$\Times$ for which $\A$ is opaque \wrt{} $\locpriv$ on the way to~$\locfinal$ is exactly~$\setn$; that is, only integer times are opaque (as the system can only leave $\locpriv$ and hence enter~$\locfinal$ at an integer time).
\begin{figure}[h]
	{\centering
		\begin{tikzpicture}[->, >=stealth', auto, node distance=2cm, thin]

		\node[location, initial] at (0, 0) (s0) {$\locinit$};
		\node[location, private] at (1.5,1.5) (sp) {$\locpriv$};
		\node[location, final] at (3,0) (s1) {$\locfinal$};

		\path (s0) edge node[align=center]{} (s1);
		\path (s0) edge node[align=center]{$\styleclock{\clock} = 0$} (sp);
		\path (sp) edge[loop above] node[align=center]{$\styleclock{\clock} = 1$\\$\styleclock{\clock} \assign 0$} (sp);
		\path (sp) edge[] node[align=center]{$\styleclock{\clock} = 0$} (s1);

		\end{tikzpicture}

	}
	\caption{TA for which the set of opaque execution times is $\setn$}
	\label{figure:opaqueforN}
\end{figure}
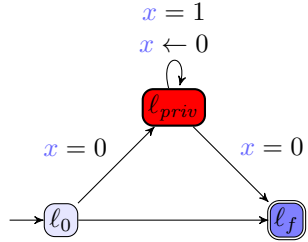

\smallskip

The synthesis counterpart allows for a higher-level problem by also synthesizing the internal timings guaranteeing opacity.

\smallskip

\defProblem
	{Timed opacity synthesis}
	{A PTA~$\A$, a private location~$\locpriv$,
		a target location~$\locfinal$
		}
	{Synthesize the parameter valuations~$\pval$ and the execution times~$\Times_\pval$ such that $\valuate{\A}{\pval}$ is opaque \wrt{} $\locpriv$ on the way to~$\locfinal$ for these execution times~$\Times_\pval$.}

\medskip

Note that the execution times can depend on the parameter valuations.

We can also define the full timed opacity decision problem, which consists in answering whether a timed automaton is fully opaque \wrt{} a private location.%

\defProblem{Full timed opacity decision}
{A TA~$\valuate{\A}{\pval}$, a private location $\locpriv$, a target location $\locfinal$}
{Is $\valuate{\A}{\pval}$ fully opaque \wrt{} $\locpriv$ on the way to~$\locfinal$?}

\smallskip

Note that a last problem of interest would be \emph{full timed opacity synthesis}, aiming at synthesizing (ideally the entire set of) parameter valuations~$\pval$ for which $\valuate{\A}{\pval}$ is \emph{fully opaque} \wrt{} $\locpriv$ on the way to~$\locfinal$.
This is left as future work.

\section{Timed opacity problems for timed automata}\label{section:TA}

In this section, we address the non-parametric problems defined in \cref{ss:problems}, \ie{} the timed opacity computation problem (\cref{ss:timed-opacity-computation}) and full timed opacity decision problem (\cref{ss:full-timed-opacity}).
We show that both problems can be solved using a construction of the $\mathit{DReach}$ sets based on the RA arithmetic~\cite{W99} (\cref{ss:DReach-RA}).

In this section, let $\A$ denote a (non-parametric) timed automaton.

\subsection{Computing $\PrivDurReach{\A}{\locpriv}{\locfinal}$ and $\PubDurReach{\A}{\locpriv}{\locfinal}$}\label{ss:DReach-RA}

We must be able to express the set of execution times, \ie{} the durations of all runs from the initial location to a given target location.
While the problem of expressing the set of execution times seems very natural for timed automata, it was barely addressed in the literature, with the exception of~\cite{BDR08,Rosenmann19}.

\subsubsection{The RA arithmetic}

We use the RA arithmetic, which is the set of first-order formulae, interpreted over the real numbers, of $\left\langle \R, +, <, \setn, 0, 1\right\rangle $ where $\setn$ is a unary predicate such that $\setn(z)$ is true iff $z$ is a natural number.
This arithmetic has a decidable theory with a complexity of 3EXPTIME~\cite{W99}.

\subsubsection{Computing execution times of timed automata}
With $\region,\region'\in\Regions_\A$, we denote by $\durRegions{\region}{\region'}$ the set of durations $d$ such that there exists a finite path $\rho = (s_i)_{i}$ in $\TTS_{\A}$ %
such that $\duration(\rho)=d$ and the associated path $\pi(\rho) = (r_k)_{0\leq k\leq K}$ in the region graph satisfies $r_0 = \region, r_K = \region'$.
It is shown in~\cite[Proposition 5.3]{BDR08} that these sets $\durRegions{\region}{\region'}$ can be defined in RA arithmetic.
Moreover, they are definable by a disjunction of terms of the form $d=m$, ${\exists z, \setn(z) \land d = m+cz}$ and ${\exists z, \setn(z)\land m+cz < d < m+cz+1}$, where $c,m\in \setn$.
Let us give the main idea of the proof presented in~\cite{BDR08} (even though this explanation is not necessary to follow our reasoning).
The idea of the proof of~\cite{BDR08} is to consider a TA~$\A^0$ obtained from~$\A$ by adding a new clock~$\clock_0$ which is reset to~0 each time it reaches the value~1 and to count all of the resets of~$\clock_0$.
The construction of~$\A^0$ ensures that each (finite) run $\rho$ of $\TTS_{\A}$ corresponds to a run $\rho^0$ of $\TTS_{\A^0}$ (at each state, the value of $\clock_0$ is the fractional part of the total time elapsed), and conversely (erasing $\clock_0$). %
The authors propose next a classical automaton~$C$ as a particular subgraph of the region graph $\Regions_{\A^0}$, where the only action~$a$ denotes the reset of $\clock_0$ (all other transitions are labeled with the silent action).
The conclusion follows because $t\in \durRegions{\region}{\region'}$ if $\integralp{t}$ is the length of a path in $C'$, the deterministic automaton obtained from~$C$ by subset construction.
\begin{lemma}[Reachability-duration computation]
	\label{lemma:DReachComputation}
	The sets $\PrivDurReach{\A}{\locpriv}{\locfinal}$ and $\PubDurReach{\A}{\locpriv}{\locfinal}$ are computable and definable in RA arithmetic.
\end{lemma}
	\begin{proof}
		Let \A~be a TA. We aim at reducing the computation of the sets $\PrivDurReach{\A}{\locpriv}{\locfinal}$ and $\PubDurReach{\A}{\locpriv}{\locfinal}$ to the computation of sets $\durRegions{\region}{\region'}$.

		First, let us compute $\PrivDurReach{\A}{\locpriv}{\locfinal}$.
		From~\A, we define a TA $\A'$ by adding a Boolean discrete variable~$\bflag$, initially $\BFalse$.
		Recall that discrete variables over a finite domain are syntactic sugar for \emph{locations}: therefore, $\locfinal$ with $\bflag = \BFalse$ and $\locfinal$ with $\bflag = \BTrue$ can be seen as two different locations.
		Then, we set $\bflag \assign \BTrue$ on any transition whose target location is~$\locpriv$; therefore, $\bflag = \BTrue$ denotes that $\locpriv$ has been visited.
		We denote by ${\locfinal'}_{\BTrue}$ the final state of $\A'$ where $\bflag = \BTrue$.
		$\PrivDurReach{\A}{\locpriv}{\locfinal}$ is exactly the set of executions times in $\A'$ between $\locinit$ and ${\locfinal'}_{\BTrue}$.
		For all the regions $\region'_i$ associated to ${\locfinal'}_{\BTrue}$, we can compute (using~\cite{BDR08}) $\durRegions{\region}{\region'_i}$, where $\region$ is the region associated to $\locinit$ in~$\A'$.
		Therefore, $\PrivDurReach{\A}{\locpriv}{\locfinal}$ can be computed as the union of all the $\durRegions{\region}{\region'_i}$ (of which there is a finite number), which is definable in RA arithmetic.%

		Second, let us compute $\PubDurReach{\A}{\locpriv}{\locfinal}$.
		We define another TA $\A''$ obtained from $\A$ by deleting all the transitions leading to $\locpriv$.
		Therefore, the set of durations reaching~$\locfinal$ in~$\A''$ is exactly the set of durations reaching~$\locfinal$ in~$\A$ associated to runs not visiting~$\locpriv$.
		$\PubDurReach{\A}{\locpriv}{\locfinal}$ is exactly the set of executions times in $\A''$ between $\locinit$ and ${\locfinal}$.
		For all the regions $\region'_i$ associated to $\locfinal$, we can compute $\durRegions{\region}{\region'_i}$, where $\region$ is the region associated to $\locinit$ in $\A''$.
		Therefore, $\PrivDurReach{\A}{\locpriv}{\locfinal}$ can be computed as the union of all the $\durRegions{\region}{\region'_i}$.
	\end{proof}
\subsection{Answering the timed opacity computation problem}\label{ss:timed-opacity-computation}
\begin{proposition}[timed opacity computation]\label{proposition:ET-opacity-computation}
	The timed opacity computation problem is solvable for TAs.
\end{proposition}
	\begin{proof}
Let $\A$~be a TA. From \cref{lemma:DReachComputation}, we can compute and define in RA arithmetic the sets $\PubDurReach{\A}{\locpriv}{\locfinal}$ and $\PrivDurReach{\A}{\locpriv}{\locfinal}$.

By the decidability of RA arithmetic, the intersection of these sets is computable.
Then, the set $\Times = \PrivDurReach{\A}{\locpriv}{\locfinal} \cap \PubDurReach{\A}{\locpriv}{\locfinal}$ is effectively computable (with a 3EXPTIME upper bound).
	\end{proof}

This positive result can be put in perspective with the negative result of~\cite{Cassez09} that proves that it is undecidable whether a TA (and even the more restricted subclass of event-recording automata~\cite{AFH99}) is opaque, in a sense that the attacker can deduce some actions, by looking at observable actions together with their timing.
The difference in our setting is that only the global time is observable, which can be seen as a single action, occurring once only at the end of the computation.
In other words, our attacker is less powerful than the attacker in~\cite{Cassez09}.

\subsection{Checking for full timed opacity}\label{ss:full-timed-opacity}
\begin{proposition}[full timed opacity decision]\label{proposition:ET-full-opacity-decsion}
	The full timed opacity decision problem is decidable for TAs.
\end{proposition}

	\begin{proof}
		Let $\A$~be a TA. From \cref{lemma:DReachComputation}, we can compute and define in RA arithmetic the sets $\PubDurReach{\A}{\locpriv}{\locfinal}$ and $\PrivDurReach{\A}{\locpriv}{\locfinal}$.

		From the decidability of RA arithmetic, the equality between these sets is decidable.
		Therefore, $ \PrivDurReach{\A}{\locpriv}{\locfinal} \overset{?}{=} \PubDurReach{\A}{\locpriv}{\locfinal}$ is decidable.
\end{proof}
From \cite{W99} and \cite[Theorem 7.5]{BDR08}, the computation of a set $\durRegions{\region}{\region'}$ is 2EXPTIME and the RA arithmetic has a decidable theory with complexity 3EXPTIME. Therefore, our construction is 3EXPTIME, which is an upper-bound for the problem complexity.
Note that, as in~\cite{BDR08}, we did not compute a lower bound for the complexity of \cref{proposition:ET-opacity-computation,proposition:ET-full-opacity-decsion}.
This remains to be done as future work.

\begin{example}
	Consider again the PTA~$\A$ in \cref{figure:example-PTA}, and let $\pval$ be such that $\pval(\param_1) = 1$ and $\pval(\param_2) = 2$.
	Recall from \cref{example:running:DurReach} that
		$\PrivDurReach{\valuate{\A}{\pval}}{\loc_2}{\loc_1} = [1, 3]$
	and
	$\PubDurReach{\valuate{\A}{\pval}}{\loc_2}{\loc_1} = [2, 3]$.
	Thus, $\PrivDurReach{\valuate{\A}{\pval}}{\loc_2}{\loc_1} \neq \PubDurReach{\valuate{\A}{\pval}}{\loc_2}{\loc_1}$ and therefore $\valuate{\A}{\pval}$ is \emph{not} \emph{(fully)} opaque \wrt{} $\loc_2$ on the way to~$\loc_1$.

	Now, consider $\pval'$ such that $\pval'(\param_1) = \pval'(\param_2) = 1.5$.
	This time, $\PrivDurReach{\valuate{\A}{\pval'}}{\loc_2}{\loc_1} = \PubDurReach{\valuate{\A}{\pval'}}{\loc_2}{\loc_1} = [1.5, 3]$ and therefore $\valuate{\A}{\pval'}$ is \emph{(fully)} opaque \wrt{} $\loc_2$ on the way to~$\loc_1$.
\end{example}
\section{The theory of parametric timed opacity \wrt{} execution times}\label{section:theory}

We address in this section the parametric problems of timed opacity \wrt{} execution times.

Let us consider the following decision problem, \ie{} the problem of checking the \emph{emptiness} of the parameter valuations and execution times set guaranteeing timed opacity \wrt{} execution times.
The decision problem associated to \emph{full} timed opacity will be considered in \cref{section:theoryFullOpacity}.

\smallskip

\defProblem
	{Timed opacity emptiness}
	{A PTA~$\A$, a private location~$\locpriv$,
		a target location~$\locfinal$
		}
	{Is the set of valuations~$\pval$ such that $\valuate{\A}{\pval}$ is opaque \wrt{} $\locpriv$ on the way to~$\locfinal$ for a non-empty set of execution times empty?}

The negation of the timed opacity emptiness consists in deciding whether there exists at least one parameter valuation for which $\valuate{\A}{\pval}$ is opaque for at least some execution time.
\label{newtext:dual}

\subsection{Undecidability in general}

We prove undecidability results for a ``sufficient'' number of clocks and parameters.
Put it differently, our proofs of undecidability require a minimum number of clocks and parameters to work; the problems are obviously undecidable for larger numbers, but the decidability is open for smaller numbers.
This will be briefly discussed in \cref{section:conclusion}.
	\label{newtext:nbclocks}

With the rule of thumb that all non-trivial decision problems are undecidable for general PTAs~\cite{Andre19STTT}, the following result is not surprising, and follows from the undecidability of reachability-emptiness for PTAs.

\begin{theorem}[undecidability]\label{proposition:TOE-undecidability}
	The timed opacity emptiness problem is undecidable for general PTAs.
\end{theorem}
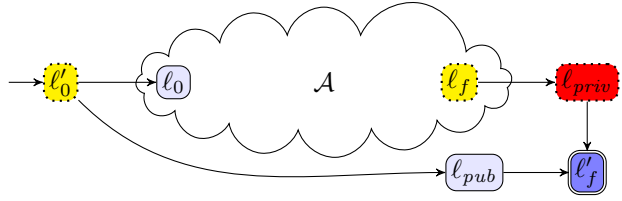
\begin{figure}[tb]
	{\centering
	\begin{tikzpicture}[->, >=stealth', auto, node distance=2cm, thin]

		\node[location] (l0) at (-2, 0) {$\locinit$};
		\node[location, urgent] (lf) at (+1.8, 0) {$\locfinal$};
		\node[cloud, cloud puffs=15.7, cloud ignores aspect, minimum width=5cm, minimum height=2cm, align=center, draw] (cloud) at (0cm, 0cm) {$\A$};

		\node[location, urgent, initial] (l0') at (-3.5, 0) {$\locinit'$};
		\node[location, urgent, private] (lpriv) at (+3.5, 0) {$\locpriv$};
		\node[location] (lpub) at (+2, -1.2) {$\locpub$};
		\node[location, final] (lf') at (+3.5, -1.2) {$\locfinal'$};

		\path
			(l0') edge (l0) %
			(lf) edge (lpriv) %
			(lpriv) edge (lf') %
			(l0') edge[out=-45,in=180] (lpub)
			(lpub) edge (lf')
			;

	\end{tikzpicture}

	}
	\caption{Reduction from reachability-emptiness} %
	\label{figure:undecidability}
\end{figure}
\begin{proof}
	We reduce from the reachability-emptiness problem, \ie{} the existence of a parameter valuation for which there exists a run reaching a given location in a PTA, which is undecidable (\eg{}~\cite{AHV93,Miller00,Doyen07,JLR15,BBLS15}).
	Consider an arbitrary PTA~$\A$ with initial location~$\locinit$ and a given location~$\locfinal$.
	It is undecidable whether there exists a parameter valuation for which there exists a run reaching~$\locfinal$%
	(proofs of undecidability in the literature generally reduce from the halting problem of a 2-counter machine which is undecidable~\cite{Minsky67}, so one can see~$\A$ as an encoding of a 2-counter machine).

	Now, add the following locations and transitions (all unguarded) as in \cref{figure:undecidability}:
		a new urgent\footnote{%
			Where time cannot elapse (depicted in dotted yellow in our figures).
		} initial location~$\locinit'$ with outgoing transitions to~$\locinit$ %
			and to a new location~$\locpub$;
		a new urgent location $\locpriv$ with an %
			incoming transition from~$\locfinal$; %
		a new final location~$\locfinal'$ with incoming transitions from~$\locpriv$ %
			and~$\locpub$. %
	Also, $\locfinal$ is made urgent.
	Let~$\A'$ denote this new PTA.

	First note that, due to the unguarded transitions, $\locfinal'$ is reachable for any parameter valuation and for any execution time by runs going through~$\locpub$ and not going through~$\locpriv$.
	That is, for all~$\pval$, $\PubDurReach{\valuate{\A'}{\pval}}{\locpriv}{\locfinal'} = [0, \infty)$.

	Assume there exists some parameter valuation~$\pval$ such that $\locfinal$ is reachable from~$\locinit$ in~$\valuate{\A}{\pval}$ for some execution times~$\Times$: then, due to our construction with additional urgent locations, $\locpriv$ is reachable on the way to $\locfinal'$ in~$\valuate{\A'}{\pval}$ for the exact same execution times~$\Times$.
	Therefore, $\valuate{\A}{\pval}$ is opaque \wrt{} $\locpriv$ on the way to~$\locfinal'$ for execution times~$\Times$.

	Conversely, if $\locfinal$ is not reachable from~$\locinit$ in~$\A$ for any valuation, then $\locpriv$ is not reachable on the way to $\locfinal'$ for any valuation in~$\A'$.
	Therefore, there is no valuation~$\pval$ such that $\valuate{\A}{\pval}$ is opaque \wrt{} $\locpriv$ on the way to~$\locfinal'$ for any execution time.
	Therefore, there exists a valuation~$\pval$ such that $\valuate{\A}{\pval}$ is opaque \wrt{} $\locpriv$ on the way to~$\locfinal'$ iff $\locfinal$ is reachable in~$\A$---which is undecidable.
\end{proof}
\begin{remark}\label{remark:nbvariables}
	Our proof reduces from the reachability-emptiness problem, for which several undecidability proofs were proposed (notably~\cite{AHV93,Miller00,Doyen07,JLR15,BBLS15}), with various flavors (numbers of parameters, integer- or dense-time, integer- or rational-valued parameters, etc.).
	See \cite{Andre19STTT} for a survey.
	Notably, this means (from~\cite{BBLS15}) that \cref{proposition:TOE-undecidability} holds for PTAs with at least 3 clocks and a single parameter.
\end{remark}

\subsection{A decidable subclass}\label{ss:decLUPTA}

We now show that the timed opacity emptiness problem is decidable for the subclass of PTAs called L/U-PTAs~\cite{HRSV02}.
Despite early positive results for L/U-PTAs~\cite{HRSV02,BlT09}, more recent results (notably \cite{JLR15,ALime17,ALR18FORMATS,ALM20}) mostly proved undecidable properties of L/U-PTAs, and therefore this positive result is welcome.

\begin{theorem}[decidability]\label{proposition:decidability-LU}
	The timed opacity emptiness problem is decidable for L/U-PTAs.
\end{theorem}
\begin{proof}
	We reduce to the timed opacity computation problem of a given TA, which is decidable (\cref{proposition:ET-opacity-computation}).

	Let $\A$ be an L/U-PTA.
	Let $\Azeroinf$ denote the structure obtained as follows: any occurrence of a lower-bound parameter is replaced with~0, and any occurrence of a conjunct $\clock \compOpLeq \param$ (where $\param$ is necessarily an upper-bound parameter) is deleted, \ie{} replaced with~$\CTrue$.

	Let us show that
		the set of valuations~$\pval$ such that $\valuate{\A}{\pval}$ is opaque \wrt{} $\locpriv$ on the way to~$\locfinal$ for a non-empty set of execution times is non empty
		iff
		the solution to the timed opacity computation problem for~$\Azeroinf$ is non-empty.

	\begin{itemize}
		\item[$\Rightarrow$]
			Assume there exists a valuation~$\pval$ such that $\valuate{\A}{\pval}$ is opaque \wrt{} $\locpriv$ on the way to~$\locfinal$ for a non-empty set of execution.
			Therefore, the solution to the timed opacity computation problem for~$\Azeroinf$ is non-empty.
			That is, there exists a duration~$d$ such that
				there exists a run of duration~$d$ such that $\locpriv$ is reachable on the way to~$\locfinal$,
			and
				there exists a run of duration~$d$ such that $\locpriv$ is unreachable on the way to~$\locfinal$.

			We now need the following monotonicity property of L/U-PTAs:

			\begin{lemma}[\cite{HRSV02}]\label{lemma:HRSV02:prop4.2}
				Let~$\A$ be an L/U-PTA and~$\pval$ be a parameter valuation.
				Let $\pval'$ be a valuation such that
				for each upper-bound parameter~$\param^u$, $\pval'(\param^u) \geq \pval(\param^u)$
				and
				for each lower-bound parameter~$\param^l$, $\pval'(\param^l) \leq \pval(\param^l)$.
				Then any run of~$\valuate{\A}{\pval}$ is a run of $\valuate{\A}{\pval'}$.
			\end{lemma}

			Therefore, from \cref{lemma:HRSV02:prop4.2}, the runs of~$\valuate{\A}{\pval}$ of duration~$d$ such that $\locpriv$ is reachable (resp.\ unreachable) on the way to~$\locfinal$
				are also runs of~$\Azeroinf$.
			Therefore, there exists a non-empty set of durations such that $\Azeroinf$ is opaque, \ie{} solution to the timed opacity computation problem for~$\Azeroinf$ is non-empty.

		\item[$\Leftarrow$]
			Assume the solution to the timed opacity computation problem for~$\Azeroinf$ is non-empty.
			That is, there exists a duration~$d$ such that
				there exists a run of duration~$d$ such that $\locpriv$ is reachable on the way to~$\locfinal$ in~$\Azeroinf$,
			and
				there exists a run of duration~$d$ such that $\locpriv$ is unreachable on the way to~$\locfinal$ in~$\Azeroinf$.

			The result could follow immediately---if only assigning $0$ and~$\infty$ to parameters was a proper parameter valuation.
			From~\cite{HRSV02,BlT09}, if a location is reachable in the TA obtained by valuating lower-bound parameters with~0 and upper-bound parameters with~$\infty$, then there exists a sufficiently large constant~$C$ such that this run exists in $\valuate{\A}{\pval}$ such that $\pval$ assigns 0 to lower-bound and~$C$ to upper-bound parameters.
			Here, we can trivially pick~$d + 1$, as any clock constraint $\clock \leq d + 1$ or $\clock < d + 1$ will be satisfied for a run of duration~$d $.
			Let $\pval$ assign 0 to lower-bound and~$d$ to upper-bound parameters.
			Then,
				there exists a run of duration~$d$ such that $\locpriv$ is reachable on the way to~$\locfinal$ in~$\valuate{\A}{\pval}$,
			and
				there exists a run of duration~$d$ such that $\locpriv$ is unreachable on the way to~$\locfinal$ in~$\valuate{\A}{\pval}$.
			Therefore, the set of valuations~$\pval$ such that $\valuate{\A}{\pval}$ is opaque \wrt{} $\locpriv$ on the way to~$\locfinal$ for a non-empty set of execution times is non empty---which concludes the proof.
			\qedhere
	\end{itemize}
\end{proof}
\begin{remark}
	The class of L/U-PTAs is known to be relatively meaningful, and many case studies from the literature fit into this class, including case studies proposed even before this class was defined in~\cite{HRSV02}, \eg{} a toy railroad crossing model and a model of Fischer mutual exclusion protocol given in~\cite{AHV93} (see~\cite{Andre19STTT} for a survey).
	Even though the PTA in \cref{figure:example-Java:PTA} does not fit in this class, it can easily be transformed into an L/U-PTA, by duplicating $\styleparam{\param}$ into $\styleparam{\param^l}$ (used in lower-bound comparisons with clocks) and $\styleparam{\param^u}$ (used in upper-bound comparisons with clocks).
\end{remark}
\subsection{Intractability of synthesis for L/U-PTAs}\label{ss:intractability}

Even though the \emph{emptiness} problem for the timed opacity \wrt{} a set of execution times $\Times$ is decidable for L/U-PTAs (\cref{proposition:decidability-LU}), the \emph{synthesis} of the parameter valuations remains intractable in general, as shown in the following \cref{proposition:intractabilitySynthTO}.
By intractable, we mean more precisely that the solution, if it can be computed, cannot (in general) be represented using any formalism for which the emptiness of the intersection with equality constraints is decidable.
That is, a formalism in which it is decidable to answer the question ``is the computed solution intersected with an equality test between variables empty?''\ cannot be used to represent the solution.
By empty, we mean emptiness of the valuations set.
For example, let us question whether we could represent the solution of the timed opacity synthesis problem for L/U-PTAs using the formalism of a \emph{finite union of polyhedra}: testing whether a finite union of polyhedra intersected with ``equality constraints'' (typically $\param_1 = \param_2$) is empty or not \emph{is} decidable.
The Parma polyhedra library~\cite{BHZ08} can typically compute the answer to this question.
Therefore, from the following \cref{proposition:intractabilitySynthTO}, finite unions of polyhedra cannot be used to represent the solution of the timed opacity synthesis problem for L/U-PTAs.
As finite unions of polyhedra are a very common formalism (not to say the \emph{de facto} standard) to represent the solutions of various timing parameters synthesis problems, the synthesis is then considered to be infeasible in practice, or \emph{intractable} (following the vocabulary used in \cite[Theorem~2]{JLR15}).

\begin{proposition}[intractability]\label{proposition:intractabilitySynthTO}
	If it can be computed, the solution to the timed opacity synthesis problem for L/U-PTAs cannot in general be represented using any formalism for which the emptiness of the intersection with equality constraints is decidable.
\end{proposition}
	\begin{proof}
		We reuse a reasoning inspired by \cite[Theorem~2]{JLR15}, and we reduce from the undecidability of the timed opacity emptiness problem for general PTAs (\cref{proposition:TOE-undecidability}).
		Assume the solution of the timed opacity synthesis problem for L/U-PTAs can be represented in a formalism for which the emptiness of the intersection with equality constraints is decidable.

		Assume an arbitrary PTA~$\A$ with notably two locations $\locpriv$ and $\locfinal$. From~$\A$, we define an L/U-PTA~$\A'$ as follow: for each parameter $p_i$ that is used both as an upper bound and a lower bound, replace its occurrences as upper bounds by a fresh parameter $p_i^u$ and its occurrences as lower bounds by a fresh parameter $p_i^l$.

		By assumption, the solution of the synthesis problem $\Gamma = \Set{(\pval,\Times_\pval)}{\valuate{\A'}{\pval} \text{ is opaque \wrt{} \locpriv\ on the way to \locfinal\ for times } \Times_\pval }$ for~$\A'$ can be computed and represented in a formalism for which the emptiness of the intersection with equality constraints is decidable.

		However, solving the emptiness of $\Set{(\pval,\Times_\pval)\in\Gamma}{\bigwedge_i \valuate{p_i^u}{\pval} = \valuate{p_i^l}{\pval}}$ (which is decidable by assumption), we can decide the timed opacity emptiness for the PTA~$\A$ (which is undecidable from \cref{proposition:TOE-undecidability}).
		This leads to a contradiction, and therefore the solution of the timed opacity synthesis problem for L/U-PTAs cannot in general be represented in a formalism for which the emptiness of the intersection with equality constraints is decidable.
	\end{proof}
\section{The theory of parametric full timed opacity}\label{section:theoryFullOpacity}

We address here the following decision problem, which asks about the emptiness of the parameter valuations set guaranteeing full timed opacity.

\smallskip

\defProblem
{Full timed opacity Emptiness}
{A PTA~$\A$, a private location~$\locpriv$,
	a target location~$\locfinal$
}
{Is the set of valuations~$\pval$ such that $\valuate{\A}{\pval}$ is fully opaque \wrt{} $\locpriv$ on the way to~$\locfinal$ empty?}

Equivalently, we are interested in deciding whether there exists at least one parameter valuation for which $\valuate{\A}{\pval}$ is fully opaque.

\subsection{Undecidability in general}

Considering that \cref{proposition:TOE-undecidability} shows that the undecidability of the timed opacity emptiness problem, the undecidability of the emptiness problem for the full timed opacity is not surprising, but does not follow immediately.

To prove this result (that will be stated formally in \cref{theorem:FTOE-undecidability}), we first need the following lemma stating that the reachability-emptiness (hereafter sometimes referred to as ``EF-emptiness'') problem is undecidable in constant time, \ie{} for a fixed time bound~$T$.
That is, the following lemma shows that, given a PTA~$\A$, a target location~$\locTarget$ and a time bound~$T$, it is undecidable whether the set of parameter valuations for which there exists a run reaching~$\locTarget$ in exactly $T$~time units is empty or not.

\begin{lemma}[Reachability in constant time]\label{lemma:undecReachFixedTime}
		The reachability-emptiness problem in constant time is undecidable for PTAs with 4~clocks and 2~parameters.
\end{lemma}
\begin{figure}[tb]
	{\centering
		\begin{tikzpicture}[->, >=stealth', auto, node distance=2cm, thin]

		\node[location, initial] (l0) at (-2, 0) {$\locinit$};
		\node[location] (lf) at (+1.8, 0) {$\locfinal$};
		\node[cloud, cloud puffs=15.7, cloud ignores aspect, minimum width=5cm, minimum height=2cm, align=center, draw] (cloud) at (0cm, 0cm) {$\A$};

		\node[location, final] (lf') at (+4, 0) {$\locfinal'$};

		\path
		(lf) edge node[] {$\styleclock{\clock} = 1$} (lf')
		;

		\end{tikzpicture}

	}
	\caption{Undecidability of EF-emptiness over constant time}
	\label{figure:undecidabilityEFemptConstantTime}
\end{figure}
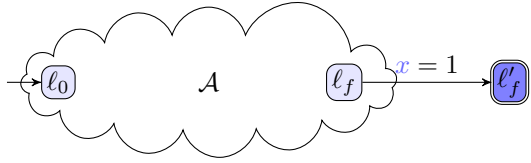
\begin{proof}
	In \cite[Theorem~3.12]{ALM20}, we showed that the EF-emptiness problem is undecidable over bounded time for PTAs with (at least) 3~clocks and 2~parameters.
	That is, given a fixed bound~$T$ and a location~$\locfinal$, it is undecidable whether the set of parameter valuations for which at least one run reaches $\locfinal$ within~$T$ time units is empty or not.

	We reduce the reachability in bounded time (\ie{} in at most~$T$ time units) to the reachability in constant time (\ie{} in exactly~$T$ time units).
	In this proof, we fix $T = 1$.

	Assume a PTA~$\A$ with a location~$\locfinal$.
	We define a PTA~$\A'$ as in \cref{figure:undecidabilityEFemptConstantTime} by
	adding a new location~$\locfinal'$, and a transition from~$\locfinal$ to~$\locfinal'$ guarded by $\clock = 1$, where $\clock$ is a new clock (initially~0), not used in~$\A$ and therefore never reset in the automaton.

	Let us show that there is no valuation such that $\locfinal$ is reachable in at most 1 time unit iff there is no valuation such that $\locfinal'$ is reachable exactly in 1 time unit.
	\begin{description}
	\item[$\Leftarrow$] Assume EF-emptiness holds in constant time for~$\A'$, \ie{} there exists no parameter valuation for which $\locfinal'$ is reachable in $T=1$ time units.
	Then, from the construction of~$\A'$, no parameter valuation exists for which $\locfinal$ is reachable in $T \leq 1$ time units.

	\item[$\Rightarrow$] Conversely, assume EF-emptiness holds over bounded time for~$\A$, \ie{} there exists no parameter valuation for which $\locfinal$ is reachable in $T \leq 1$ time units.
	Then, from the construction of~$\A'$, no parameter valuation exists for which $\locfinal'$ is reachable in $T = 1$ time units.
	\end{description}
	This concludes the proof of the lemma.
\end{proof}

We can now state and prove \cref{theorem:FTOE-undecidability}.

\begin{theorem}[undecidability]\label{theorem:FTOE-undecidability}
	The full timed opacity emptiness problem is undecidable for general PTAs with (at least) 4~clocks and 2~parameters.
\end{theorem}
\begin{proof}
		We reduce from the reachability-emptiness problem in constant time, which is undecidable (\cref{lemma:undecReachFixedTime}).

		Consider an arbitrary PTA~$\A$ with (at least) 4~clocks and 2~parameters, with initial location~$\locinit$ and a given location~$\locfinal$.
		We add the following locations and transitions in~$\A$\ to obtain a PTA~$\A'$, as in \cref{figure:undecidabilityFullOpacityPTA}:
			a new urgent initial location $\locinit'$, with outgoing transition to $\locinit$ and to a new location $\locpub$,
			a new urgent location $\locpriv$ with an incoming transition from $\locfinal$,
			a new urgent and final location $\locfinal'$ with incoming transitions from $\locpriv$\ and $\locpub$, and
			a guard $\clock = 1$ (with a new clock $\clock$, never reset) on the transition from $\locpub$ to~$\locfinal'$.

		First, note that, due to the guarded transition, $\locfinal'$ is reachable for any parameter valuation and (only) for an execution time equal to~1 by runs going through $\locpub$ and not going through~$\locpriv$.
		That is, for all $\pval$, $\PubDurReach{\valuate{\A}{\pval}}{\locpriv}{\locfinal} = \set{1}$.

		We show that
			there exists a valuation $\pval$ such that $\valuate{\A'}{\pval}$ is fully opaque \wrt{} $\locpriv$ on the way to $\locfinal'$
			iff
			there exists a valuation $\pval$ such that $\locfinal$ is reachable in $\valuate{\A}{\pval}$ for execution time equal to~1.
		\begin{itemize}
			\item[$\Leftarrow$]
			Assume there exists some valuation $\pval$ such that $\locfinal$ is reachable from $\locinit$ in~$\A$ (only) for execution time equal to~1.
			Then, due to our construction, $\locpriv$ is reachable on the way to $\locfinal'$ in $\valuate{\A'}{\pval}$ for the exact same execution time 1.
			Therefore, $\valuate{\A'}{\pval}$ is fully opaque \wrt{} $\locpriv$ on the way to $\locfinal'$.

			\item[$\Rightarrow$]
			Conversely, if $\locfinal$ is not reachable from $\locinit$ in $\A$ for any valuation for execution time 1, then $\locpriv$ is not reachable on the way to $\locfinal'$ for any valuation of $\A'$.
                Therefore, there is no valuation $\pval$ such that $\valuate{\A'}{\pval}$ is fully opaque \wrt{} $\locpriv$ on the way to $\locfinal'$ for execution time~1. %
		\end{itemize}
		Therefore, there exists a valuation $\pval$ such that $\valuate{\A'}{\pval}$ is fully opaque \wrt{} $\locpriv$ on the way to $\locfinal'$
			iff
			there exists a valuation $\pval$ such that $\locfinal$ is reachable in $\valuate{\A}{\pval}$ for execution time equal to~1---which is undecidable.
		This concludes the proof.
		
		Let us briefly discuss the minimum number of clocks necessary to obtain undecidability using  our proof (the case of smaller numbers of clocks remains open).
		Recall that \cref{lemma:undecReachFixedTime} needs 4~clocks; in the current proof of \cref{theorem:FTOE-undecidability}, we add a new clock~$\clock$ which is never reset; however, since the proof of \cref{lemma:undecReachFixedTime} also uses a clock which is never reset, therefore we can reuse it, and our proof does not need an additional clock.
		So the result holds for 4~clocks and 2~parameters.
	\end{proof}
\begin{figure}[tb]
	{\centering
		\begin{tikzpicture}[->, >=stealth', auto, node distance=2cm, thin]

		\node[location] (l0) at (-2, 0) {$\locinit$};
		\node[location, urgent] (lf) at (+1.8, 0) {$\locfinal$};
		\node[cloud, cloud puffs=15.7, cloud ignores aspect, minimum width=5cm, minimum height=2cm, align=center, draw] (cloud) at (0cm, 0cm) {$\A$};

		\node[location, urgent, initial] (l0') at (-3.5, 0) {$\locinit'$};
		\node[location, urgent, private] (lpriv) at (+4, 0) {$\locpriv$};
		\node[location] (lpub) at (+2, -1.2) {$\locpub$};
		\node[location, final] (lf') at (+4, -1.2) {$\locfinal'$};

		\path
		(l0') edge (l0) %
		(lf) edge (lpriv) %
		(lpriv) edge (lf') %
		(l0') edge[out=-45,in=180] (lpub)
		(lpub) edge node {$\styleclock{\clock} = 1$} (lf')
		;

		\end{tikzpicture}

	}
	\caption{Reduction from reachability-emptiness} %
	\label{figure:undecidabilityFullOpacityPTA}
\end{figure}
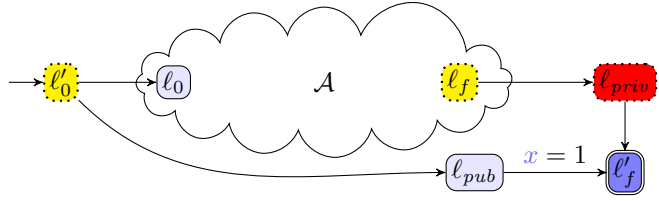
\subsection{Undecidability for L/U-PTAs}
	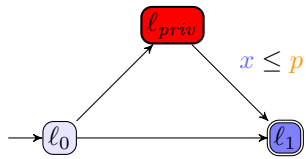
\begin{figure}[tb]
		{\centering
			\begin{tikzpicture}[->, >=stealth', auto, node distance=2cm, thin]

			\node[location, initial] at (0, 0) (s0) {$\loc_0$};
			\node[location, private] at (1.5,1.5) (sp) {$\locpriv$};
			\node[location, final] at (3,0) (s1) {$\loc_1$};

			\path (s0) edge node[align=center]{} (s1);
			\path (s0) edge node[align=center]{} (sp);
			\path (sp) edge[] node[align=center]{$\styleclock{\clock} \leq \styleparam{p}$} (s1);

			\end{tikzpicture}

		}
		\caption{$\A_{0,\infty}$ is not sufficient for the full timed opacity emptiness problem} %
		\label{figure:A0infNotSufficientForFTOEP}
	\end{figure}

	Note that reasoning like in \cref{ss:decLUPTA}, \ie{} reducing the emptiness problem to a decision problem of the non-parametric $\A_{0,\infty}$, is not relevant.
	\cref{figure:A0infNotSufficientForFTOEP} shows an L/U-PTA \A{} (and more precisely, a U-PTA, \ie{} an L/U-PTA with an empty set of lower-bound parameters~\cite{BlT09}) which is not fully opaque for any parameter valuation, but whose associated TA $\A_{0,\infty}$ is.

\begin{figure}[tb]
	{\centering
		\begin{tikzpicture}[->, >=stealth', auto, node distance=2cm, thin]

		\node[location, initial] at (0, 0) (s0) {$\loc_0$};
		\node[location, private] at (1.5,1.5) (sp) {$\locpriv$};
		\node[location, final] at (3,0) (s1) {$\locfinal$};

		\path (s0) edge node[align=center]{$\styleclock{\clock} \leq \styleparam{p}$} (s1);
		\path (s0) edge node[align=center]{} (sp);
		\path (sp) edge[] node[align=center]{$\styleclock{\clock} \leq 1$} (s1);

		\end{tikzpicture}

	}
	\caption{No monotonicity for full timed opacity in L/U-PTAs}
	\label{figure:counterexample-monotonicity-LU}
\end{figure}
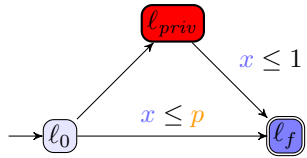

In addition, while it is well-known that L/U-PTAs enjoy a monotonicity for reachability properties (``enlarging an upper-bound parameter or decreasing a lower-bound parameter preserves reachability'') as recalled in \cref{lemma:HRSV02:prop4.2}, we can show in the following example that this is not the case for full timed opacity.

\begin{example}
	Consider the PTA in \cref{figure:counterexample-monotonicity-LU}.
	First assume $\pval$ such that $\pval(\param) = 0.5$.
    Then, $\valuate{\A}{\pval}$ is not full timed opaque: indeed, $\locfinal$ can be reached in 1 time unit via $\locpriv$, but not without going through~$\locpriv$.

	Second, assume $\pval'$ such that $\pval'(\param) = 1$.
	Then, $\valuate{\A}{\pval'}$ is full timed opaque: indeed, $\locfinal$ can be reached for any duration in $[0,1]$ by runs both going through and not going through $\locpriv$.

	Finally, let us enlarge~$\param$ further, and assume $\pval''$ such that $\pval''(\param) = 2$.
	Then, $\valuate{\A}{\pval''}$ becomes again not full timed opaque: indeed, $\locfinal$ can be reached in 2 time units not going through~$\locpriv$, but cannot be reached in 2 time units by going through~$\locpriv$.

	As a side note, remark that this PTA is actually a U-PTA, that is, monotonicity for this problem does not even hold for U-PTAs.
\end{example}

In fact, we show that, while the timed opacity emptiness problem is decidable for L/U-PTAs (\cref{proposition:decidability-LU}), the full timed opacity emptiness problem becomes undecidable for this same class (from 4~parameters).
This confirms (after previous works in~\cite{JLR15,ALime17,ALR18FORMATS}) that L/U-PTAs stand at the frontier between decidability and undecidability.

\begin{theorem}\label{theorem:FTOE:LU}
	The full timed opacity emptiness problem is undecidable for L/U-PTAs with (at least) 4~clocks and 4~parameters.
\end{theorem}
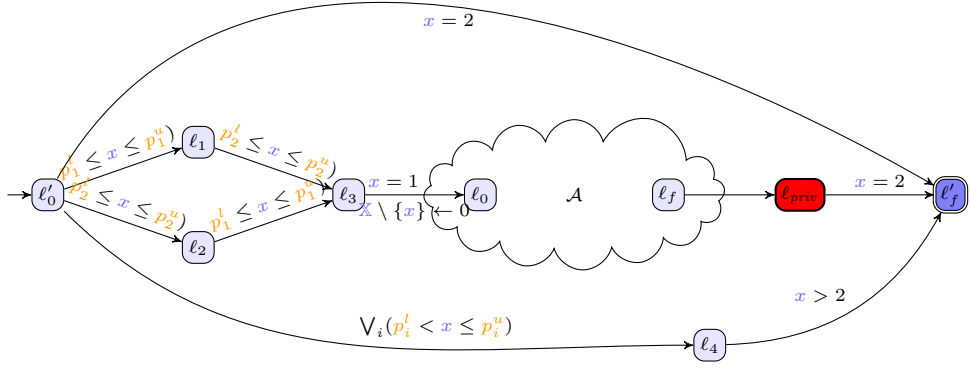
\begin{figure}[tb]
	{\centering
		\scriptsize
		\begin{tikzpicture}[->, >=stealth', auto, node distance=2cm, thin]

		\node[location] (l0) at (-1.25, 0) {$\locinit$};
		\node[cloud, cloud puffs=15.7, cloud ignores aspect, minimum width=4cm, minimum height=2cm, align=center, draw] (cloud) at (0, 0) {$\A$};
		\node[location] (lf) at (+1.25, 0) {$\locfinal$};

		\node[location, initial] (l0') at (-7, 0) {$\locinit'$};
		\node[location] (l1) at (-5, .7) {$\loc_1$};
		\node[location] (l2) at (-5, -.7) {$\loc_2$};
		\node[location] (l3) at (-3, 0) {$\loc_3$};
		\node[location, private] (lpriv) at (+3.0, 0) {$\locpriv$};
		\node[location] (l4) at (+1.8, -2) {$\loc_4$};

		\node[location, final] (lf') at (+5, 0) {$\locfinal'$};

		\path
		(lf) edge node[] {} (lpriv)
		(lpriv) edge node[] {$\clockx = 2$} (lf')

		(l0') edge node[sloped,align=center] {$\styleparam{\param_1^l} \leq \clockx \leq \styleparam{\param_1^u})$} (l1)
		(l1) edge node[sloped,align=center] {$\styleparam{\param_2^l} \leq \clockx \leq \styleparam{\param_2^u})$} (l3)
		(l0') edge node[sloped] {$\styleparam{\param_2^l} \leq \clockx \leq \styleparam{\param_2^u})$} (l2)
		(l2) edge node[sloped,align=center] {$\styleparam{\param_1^l} \leq \clockx \leq \styleparam{\param_1^u})$} (l3)
		(l3) edge node[above,align=center,xshift=-1em] {$\clockx = 1$} node[below] {$\styleclock{\Clock} \setminus \{ \clockx \} \leftarrow 0$} (l0)

		(l0') edge[out=315,in=180] node[] {$\bigvee_{i}(\styleparam{\param_i^l} < \clockx \leq \styleparam{\param_i^u})$} (l4)
		(l4) edge[bend right] node[] {$\clockx > 2$} (lf')

		(l0') edge[out=60,in=150] node[below] {$\clockx = 2$} (lf')

		;

		\end{tikzpicture}

	}
	\caption{Undecidability of full timed opacity emptiness for L/U-PTAs}
	\label{figure:undecidabilityFTOELU}
\end{figure}
\begin{proof}
	Let us recall from \cite[Theorem~3.12]{ALM20} that the EF-emptiness problem is undecidable over bounded time for PTAs with (at least) 3~clocks and 2~parameters.
	Assume a PTA~$\A$ with 3~clocks and 2~parameters, say $\param_1$ and $\param_2$, and a target location~$\locfinal$.
	Take~1 as a time bound.
	From \cite[Theorem~3.12]{ALM20}, it is undecidable whether there exists a parameter valuation for which~$\locfinal$ is reachable in time $\leq 1$.

	The idea of our proof is that, as in~\cite{JLR15,ABPP19}, we ``split'' each of the two parameters used in~$\A$ into a lower-bound parameter ($\param_1^l$ and $\param_2^l$) and an upper-bound parameter ($\param_1^u$ and $\param_2^u$).
	Each construction of the form $\clock < \param_i$ (resp.\ $\clock \leq \param_i$) is replaced with $\clock < \param_i^u$ (resp.\ $\clock \leq \param_i^u$)
	while
	each construction of the form $\clock > \param_i$ (resp.\ $\clock \geq \param_i$) is replaced with $\clock > \param_i^l$ (resp.\ $\clock \geq \param_i^l$);
	$\clock = \param_i$ is replaced with $\param_i^l \leq \clock \leq \param_i^u$.

	The idea is that the PTA~$\A$ is exactly equivalent to our construction with duplicated parameters only when $\param_1^l = \param_1^u$ and $\param_2^l = \param_2^u$.
	The crux of the rest of this proof is that we will ``rule out'' any parameter valuation not satisfying these equalities, so as to use directly the undecidability result of \cite[Theorem~3.12]{ALM20}.

	Now, consider the extension of~$\A$ given in \cref{figure:undecidabilityFTOELU}, and let~$\A'$ be this extension.
	We assume that $\clock$ is an extra clock not used in~$\A$.
	The syntax ``$\Clock \setminus \{ \clock \} \leftarrow 0$'' denotes that all clocks of the original PTA~$\A$ are reset---but not the new clock~$\clock$.
	The guard on the lower transition from $\locinit'$ to $\loc_4$ stands for 2 different transitions guarded with
		$\param_1^l < \clock \leq \param_1^u$,
		and
		$\param_2^l < \clock \leq \param_2^u$,
		respectively.
	Let us first make the following observations:
	\begin{enumerate}
		\item for any parameter valuation, one can take the upper transition from $\locinit'$ to $\locfinal'$ at time~2, \ie{} $\locfinal'$ is always reachable in time~2 without going through location~$\locpriv$;
		\item the original automaton~$\A$ can only be entered whenever $\param_1^l \leq \param_1^u$ and $\param_2^l \leq \param_2^u$; going from~$\locinit'$ to~$\locinit$ takes exactly 1 time unit (due to the $\clock = 1$ guard);
		\item if a run reaches $\locpriv$ on the way to~$\locfinal'$, then its duration is necessarily~2;
		\item from \cite[Theorem~3.12]{ALM20}, it is undecidable whether there exists a parameter valuation for which there exists a run reaching~$\locfinal$ from~$\locinit$ in time $\leq 1$, \ie{} reaching~$\locfinal$ from~$\locinit'$ in time $\leq 2$.
	\end{enumerate}

	Let us consider the following cases:
	\begin{enumerate}
		\item if $\param_1^l > \param_1^u$ or $\param_2^l > \param_2^u$, then thanks to the transitions from $\locinit'$ to~$\locinit$, there is no way to enter the original PTA~$\A$ (and therefore to reach $\locpriv$ on the way to~$\locfinal'$); since these valuations can still reach $\locfinal'$ for some execution times (notably $\clock = 2$ through the upper transition from~$\locinit'$ to~$\locfinal'$), then $\A'$ is not fully opaque for any of these valuations.

		\item if $\param_1^l < \param_1^u$ or $\param_2^l < \param_2^u$, then one of the lower transitions from~$\locinit'$ to~$\loc_4$ can be taken, and therefore $\locfinal'$ is reachable in a time $> 2$ without going through~$\locpriv$.
		Since no run can reach $\locfinal'$ while going through~$\locpriv$ for a duration $\neq 2$, then again $\A'$ is not fully opaque for any of these valuations.

		\item if $\param_1^l = \param_1^u$ and $\param_2^l = \param_2^u$, then the behavior of the modified~$\A$ (with duplicate parameters) is exactly the one of the original~$\A$.
			Also, note that the lower transitions from~$\locinit'$ to~$\locfinal'$ (via~$\loc_4$) cannot be taken.
			In contrast, the upper transition from~$\locinit'$ to~$\locfinal'$ can still be taken, and therefore there exists a run of duration~2 reaching~$\locfinal'$ without visiting~$\locpriv$.

			Now, assume there exists a parameter valuation for which there exists a run of~$\A$ of duration $\leq 1$ reaching $\locfinal$.
			And, as a consequence, $\locpriv$ is reachable, and therefore there exists some run of duration~2 (including the 1 time unit to go from~$\locinit$ to~$\locinit'$) reaching $\locfinal'$ after visiting~$\locpriv$.
			From the above reasoning, all runs reaching~$\locfinal'$ have duration~2; in addition, we exhibited a run visiting~$\locpriv$ and a run not visiting~$\locpriv$; therefore the modified automaton $\A'$ is fully opaque for such a parameter valuation.

			Conversely, assume there exists no parameter valuation for which there exists a run of~$\A$ of duration $\leq 1$ reaching $\locfinal$.
			In that case, $\A'$ is not fully opaque for any parameter valuation.
	\end{enumerate}

	As a consequence, there exists a parameter valuation~$\pval'$ for which $\valuate{\A'}{\pval'}$ is fully opaque iff there exists a parameter valuation~$\pval$ for which there exists a run in~$\valuate{\A}{\pval}$ of duration $\leq 1$ reaching $\locfinal$---which is undecidable from \cite[Theorem~3.12]{ALM20}.
\end{proof}

\section{Parameter synthesis for opacity}\label{section:synthesis}

Despite the negative theoretical result of \cref{proposition:TOE-undecidability}, we now address the timed opacity synthesis problem for the full class of PTAs.
Our method may not terminate (due to the undecidability) but, if it does, its result is correct.
Our workflow can be summarized as follows.

\begin{enumerate}
	\item We enrich the original PTA by adding a Boolean flag~$\bflag$ and a final synchronization action;
	\item We perform \emph{self-composition} (\ie{} parallel composition with a copy of itself) of this modified PTA;
	\item We perform reachability-synthesis using \EFsynth{} on~$\locfinal$ with contradictory values of~$\bflag$.
\end{enumerate}

We detail each operation in the following.

In this section, we assume a PTA~$\A$, a given private location~$\locpriv$ and a given final location~$\locfinal$.

\subsection{Enriching the PTA}\label{ss:enriching}

We first add a Boolean flag $\bflag$ initially set to $\BFalse$, and then set to $\BTrue$ on any transition whose target location is~$\locpriv$ (in the line of the proof of \cref{proposition:ET-opacity-computation}).
Therefore, $\bflag = \BTrue$ denotes that $\locpriv$ has been visited.
Second, we add a synchronization action $\actionEnd$ on any transition whose target location is~$\locfinal$.
Third, we add a new clock $\clockabs$ (never reset) together with a new parameter $\paramabs$, and we guard all transitions to~$\locfinal$ with $\clockabs = \paramabs$.
This will allow to measure the (parametric) execution time.
Let $\Enrich(\A, \locpriv, \locfinal)$ denote this procedure.

\begin{figure*}[tb]
\newcommand{\ratio}{0.5\textwidth}
 
	\centering
	\footnotesize

	\begin{tikzpicture}[scale=1, xscale=2.3, yscale=2.3, auto, ->, >=stealth']
 
		\node[location, initial] at (0,0) (l1) {$\loc_1$};
 
		\node[location] at (1, 0) (l2) {$\loc_2$};
 
		\node[location] at (1.8, 0) (l3) {$\loc_3$};
 
		\node[location] at (1.8, 1) (final1) {\styleloc{error}};
 
		\node[location] at (2.6, 0) (l4) {$\loc_4$};
 
		\node[location, private] at (3.5, 0) (hidden1) {$\locpriv$};
 
		\node[location] at (3.5, 1) (hidden2) {$\loc_5$};
 
		\node[location, final] at (4.7, 0) (final2) {$\locfinal$};
 
		\node[invariant, below=of l1] {$\clockcl \leq \styleparam{\epsilon}$};
		\node[invariant, below=of l2] {$\clockcl \leq \styleparam{\epsilon}$};
		\node[invariant, below=of l3] {$\clockcl \leq \styleparam{\epsilon}$};
		\node[invariant, below=of l4] {$\clockcl \leq \styleparam{\epsilon}$};
		
		\path (l1) edge node[align=center]{$\clockcl \leq \styleparam{\epsilon}$ \\ $\styleact{\mathrm{setupserver}}$} node[below] {$\clockcl \assign 0$} (l2);
 
		\path (l2) edge node[align=center]{$\clockcl \leq \styleparam{\epsilon}$ \\  $\styleact{\mathrm{read}?}\styledisc{x}$} node[below] {$\clockcl \assign 0$} (l3);

		\path (l3) edge node[above left, align=center]{$\clockcl \leq \styleparam{\epsilon}$ \\  $ \styledisc{x} < 0$} (final1);

		\path (l3) edge node[align=center]{$\clockcl \leq \styleparam{\epsilon}$ \\  $ \styledisc{x} \geq 0$ } node[below] {$\clockcl \assign 0$} (l4);

		\path (l4) edge node{\begin{tabular}{@{} c @{\ } c@{} }
		& $ \styledisc{x} \leq \styledisc{secret}$\\
		 $\land $ & $ \clockcl \leq \styleparam{\epsilon}$\\
		\end{tabular}} node [below, align=center]{$\clockcl\assign 0$ \\ $\styledisc{b}\assign\CTrue$} (hidden1);

		\path (l4) edge[bend left] node{\begin{tabular}{@{} c @{\ } c@{} }
		& $ \styledisc{x} > \styledisc{secret}$\\
		 $\land $ & $ \clockcl \leq \styleparam{\epsilon}$\\
		 & $\clockcl\assign 0$\\
		\end{tabular}} (hidden2);

		\path (hidden1) edge node{\begin{tabular}{@{} c @{\ } c@{} }
		& $ 32^2  \leq \clockcl$\\
		$\land$ & $ \clockcl \leq 32^2 + \styleparam{\epsilon}$\\
		$\land$ & $\styleclock{ \clockabs }= \styleparam{\paramabs}$\\
		\end{tabular}} node[below]{$\actionEnd$} (final2);

		\path (hidden2) edge[bend left] node[right,xshift=8]{\begin{tabular}{@{} c @{\ } c@{} }
		& $ \styleparam{p} \times 32^2 \leq \clockcl$\\
		$\land$ & $ \clockcl \leq \styleparam{p} \times 32^2 + \styleparam{\epsilon}$\\
		$\land$ & $\styleclock{ \clockabs }= \styleparam{\paramabs}$\\
		& $\actionEnd$\\
		\end{tabular}} (final2);
 
	\end{tikzpicture}
	\caption{Transformed version of \cref{figure:example-Java:PTA}}
	\label{figure:example-Java:PTA-transformed}

\end{figure*}
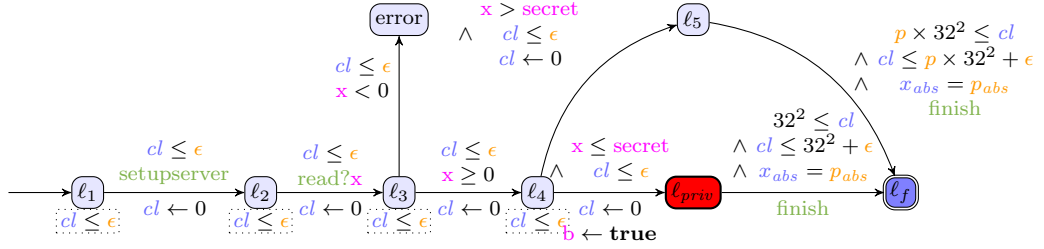

\begin{example}
	\cref{figure:example-Java:PTA-transformed} shows the transformed version of the PTA in \cref{figure:example-Java:PTA}.
\end{example}
\subsection{Self-composition}

We use here the principle of \emph{self-composition}~\cite{BDR11}, \ie{} composing the PTA with a copy of itself.
More precisely, given a PTA~$\A' = \Enrich(\A, \locpriv, \locfinal)$, we first perform an identical copy of~$\A'$ \emph{with distinct variables}: that is, a clock~$\clock$ of~$\A'$ is distinct from a clock~$\clock$ in the copy of~$\A'$---which can be trivially performed using variable renaming.\footnote{%
	In fact, the fresh clock~$\clockabs$ and parameter~$\paramabs$ can be shared to save two variables, as~$\clockabs$ is never reset, and both PTAs enter~$\locfinal$ at the same time, therefore both ``copies'' of~$\clockabs$ and $\paramabs$ always share the same values.
}
Let~$\Copy(\A')$ denote this copy of~$\A'$.
We then compute $\A' \parallel_{\{\actionEnd\}} \Copy(\A')$.
That is, $\A'$ and~$\Copy(\A')$ evolve completely independently due to the interleaving---except that they are forced to enter~$\locfinal$ at the same time, thanks to the synchronization action~$\actionEnd$.

\subsection{Synthesis}\label{ss:synthesis}

Then, we apply reachability synthesis \EFsynth{} (over all parameters, \ie{} the ``internal'' timing parameters, but also the $\paramabs$ parameter) to the following goal location:
the original~$\A'$ is in~$\locfinal$ with $\bflag = \BTrue$ while its copy~$\Copy(\A')$ is in~$\locfinal'$ with $\bflag' = \BFalse$ (primed variables denote variables from the copy).
Intuitively, we synthesize timing parameters and execution times such that there exists a run reaching~$\locfinal$ with $\bflag = \BTrue$ (\ie{} that has visited $\locpriv$) and there exists another run of same duration reaching~$\locfinal$ with $\bflag = \BFalse$ (\ie{} that has not visited $\locpriv$).

Let $\SynthOp(\A, \locpriv, \locfinal)$ denote the entire procedure.
We formalize $\SynthOp$ in \cref{algo:SynthOp}, where ``$\locfinal \land \bflag = \BTrue$'' denotes the location $\locfinal$ with $\bflag = \BTrue$.
Recall that $\paramabs$ is added by the enrichment step described in \cref{ss:enriching}.
The set of execution times $\Times$ is therefore given by the possible valuations of~$\paramabs$; these valuations may depend on the model timing parameters (in the form of a constraint).
Finally note that \EFsynth{} is called on a set made of a single location of $\A' \parallel_{\{\actionEnd\}} \Copy(\A')$; by definition of the synchronous product, this location is a \emph{pair} of locations, one from~$\A'$ (\ie{} ``$\locfinal \land \bflag = \BTrue$'') and one from~$\Copy(\A')$ (\ie{} ``$\locfinal' \land \bflag' = \BFalse$'').

\begin{algorithm}[tb]
	\Input{A PTA $\A$ with parameters set~$\Param$, locations~$\locpriv,\locfinal$}
	\Output{Parameter constraint $\K$ over $\Param \cup \{ \paramabs \}$}

	\BlankLine

	$\A' \assign \Enrich(\A, \locpriv, \locfinal)$

	$\A'' \assign \A' \parallel_{\{\actionEnd\}} \Copy(\A')$

	\Return $\EFsynth \Big(\A'', \big\{ ( \locfinal \land \bflag = \BTrue , \locfinal' \land \bflag' = \BFalse ) \big\} \Big)$

	\caption{$\SynthOp(\A, \locpriv, \locfinal)$}
	\label{algo:SynthOp}
\end{algorithm}
\begin{example}
	Consider again the PTA~$\A$ in \cref{figure:example-Java:PTA}: its enriched version~$\A'$ is given in \cref{figure:example-Java:PTA-transformed}.
	Fix $\pval(\styleparam{\epsilon}) = 1$, $\pval(\styleparam{p}) = 2$.
	We then perform the synthesis applied to the self-composition of~$\A'$ according to \cref{algo:SynthOp}.
	The result obtained with \imitator{} is:
	$\paramabs = \emptyset$ (as expected from \cref{example:Java-PTA}).

	Now fix $\pval(\styleparam{\epsilon}) = 2$, $\pval(\styleparam{p}) = 1.002$.
	We obtain:
	$\paramabs \in [1026.048, 1034]$ (again, as expected from \cref{example:Java-PTA}).

	Now let us keep all parameters unconstrained.
	The result of \cref{algo:SynthOp} is the following 3-dimensional constraint:
\begin{tabular}{l l}
 & $5 \times \styleparam{\epsilon} + 1024  \geq \styleparam{\paramabs} \geq 1024$
 \\
$ \land$ & $1024 \times \styleparam{p} + 5 \times \styleparam{\epsilon} \geq \styleparam{\paramabs} \geq 1024 \times \styleparam{p}  \geq 0$
 \end{tabular}
\end{example}
\subsection{Correctness}\label{ss:correctness}

We will state below that, whenever $\SynthOp(\A, \locpriv, \locfinal)$ terminates, then its result is an exact (sound and complete) answer to the timed opacity synthesis problem.

Let us first prove a technical lemma used later to prove the soundness of \SynthOp{}.

\begin{lemma}\label{lemma:ET}
	Assume $\SynthOp(\A, \locpriv, \locfinal)$ terminates with result~$\K$.
	For all $\pval \models \K$, there exists a run ending in~$\locfinal$ at time~$\pval(\paramabs)$ in~$\valuate{\A}{\pval}$.
\end{lemma}
\begin{proof}
	From the construction of the procedure~$\Enrich$, we added a new clock $\clockabs$ (never reset) together with a new parameter $\paramabs$, and we guarded all transitions to~$\locfinal$ with $\clockabs = \paramabs$.
	Therefore, valuations of~$\paramabs$ correspond exactly to the times at which $\locfinal$ can be reached in~$\valuate{\A}{\pval}$.
\end{proof}

We can now prove soundness and completeness.

\begin{proposition}[soundness]\label{proposition:soundness}
	Assume $\SynthOp(\A, \locpriv, \locfinal)$ terminates with result~$\K$.
	For all $\pval \models \K$, there exists a run of duration $\pval(\paramabs)$ such that $\locpriv$ is reachable on the way to~$\locfinal$ in~$\valuate{\A}{\pval}$
		and
	there exists a run of duration $\pval(\paramabs)$ such that $\locpriv$ is avoided on the way to~$\locfinal$ in~$\valuate{\A}{\pval}$.
\end{proposition}
\begin{proof}
	$\SynthOp(\A, \locpriv, \locfinal)$ is the result of \EFsynth{} called on the self-composition of~$\Enrich(\A, \locpriv, \locfinal)$.
	Recall that $\Enrich$ has enriched~$\A$ with the addition of a guard $\clockabs = \paramabs$ on the incoming transitions of~$\locfinal$, as well as a Boolean flag $\bflag$ that is $\BTrue$ iff $\locpriv$ was visited along a run.
	Assume $\pval \models \K$.
	From \cref{prop:EFsynth}, there exists a run of $\A''$ reaching $ \locfinal \land \bflag = \BTrue , \locfinal' \land \bflag' = \BFalse$.
	From \cref{lemma:ET}, this run
	takes $\pval(\paramabs)$ time units.
	From the self-composition that is made of interleaving only (except for the final synchronization), there exists a run of duration $\pval(\paramabs)$ such that $\locpriv$ is reachable on the way to~$\locfinal$ in~$\valuate{\A}{\pval}$
	and
	there exists a run of duration $\pval(\paramabs)$ such that $\locpriv$ is avoided on the way to~$\locfinal$ in~$\valuate{\A}{\pval}$.
\end{proof}
\begin{proposition}[completeness]\label{proposition:completeness}
	Assume $\SynthOp(\A, \locpriv, \locfinal)$ terminates with result~$\K$.
	Assume $\pval$.
	Assume there exists a run of duration $\pval(\paramabs)$ such that $\locpriv$ is reachable on the way to~$\locfinal$ in~$\valuate{\A}{\pval}$
		and
	there exists a run of duration $\pval(\paramabs)$ such that $\locpriv$ is avoided on the way to~$\locfinal$ in~$\valuate{\A}{\pval}$.
	Then $\pval \models \K$.
\end{proposition}
\begin{proof}
	Assume $\SynthOp(\A, \locpriv, \locfinal)$ terminates with result~$\K$.
	Assume $\pval$.
	Assume there exists a run~$\varrun$ of duration $\pval(\paramabs)$ such that $\locpriv$ is reachable on the way to~$\locfinal$ in~$\valuate{\A}{\pval}$
		and
	there exists a run~$\varrun'$ of duration $\pval(\paramabs)$ such that $\locpriv$ is avoided on the way to~$\locfinal$ in~$\valuate{\A}{\pval}$.

	First, from~$\Enrich$, there exists a run~$\varrun$ of duration $\pval(\paramabs)$ such that $\locpriv$ is reachable (resp.\ avoided) on the way to~$\locfinal$ in~$\valuate{\A}{\pval}$ implies that there exists a run~$\varrun$ of duration $\pval(\paramabs)$ such that $\locfinal \land \bflag = \BTrue$ (resp.\ $\bflag = \BFalse$) is reachable in~$\valuate{\Enrich(\A)}{\pval}$.

	Since our self-composition allows any interleaving, runs~$\varrun$ of $\valuate{\A'}{\pval}$ and~$\varrun'$ in~$\valuate{\Copy(\A')}{\pval}$ are independent---except for reaching~$\locfinal$.
	Since $\varrun$ and $\varrun'$ have the same duration $\pval(\paramabs)$, then they both reach $\locfinal$ at the same time and, from our definition of self-composition, they can simultaneously fire action~$\actionEnd$ and enter~$\locfinal$ at time~$\pval(\paramabs)$.
	Hence, there exists a run reaching $\locfinal \land \bflag = \BTrue , \locfinal' \land \bflag' = \BFalse$ in~$\valuate{\A''}{\pval}$.

	Finally, from \cref{prop:EFsynth}, $\pval \models \K$.
\end{proof}
\begin{theorem}[correctness]\label{theorem:correctness}
	Assume $\SynthOp(\A, \locpriv, \locfinal)$ terminates with result~$\K$.
	Assume $\pval$.
	The following two statements are equivalent:
	\begin{enumerate}
		\item
			There exists a run of duration $\pval(\paramabs)$ such that $\locpriv$ is reachable on the way to~$\locfinal$ in~$\valuate{\A}{\pval}$
				and
			there exists a run of duration $\pval(\paramabs)$ such that $\locpriv$ is avoided on the way to~$\locfinal$ in~$\valuate{\A}{\pval}$.
		\item $\pval \models \K$.
	\end{enumerate}
\end{theorem}
\begin{proof}
	From \cref{proposition:soundness,proposition:completeness}
\end{proof}

\section{Experiments}\label{section:experiments}
\subsection{Experimental environment}\label{ss:exp-env}

We use \imitator{}~\cite{Andre21}, a parametric timed model checking tool taking as input networks of PTAs extended with several handful features such as shared global discrete variables, PTA synchronization through strong broadcast, non-timing rational-valued parameters, etc.
\imitator{} supports various parameter synthesis algorithms, including reachability synthesis.
\imitator{} represents symbolic states as polyhedra, relying on PPL~\cite{BHZ08}.
\imitator{} is a leading tool for parameter synthesis for extensions of parametric timed automata.
Related tools are
Romeo~\cite{LRST09} (which cannot be used here, as it does not support parametric timed automata, but extensions of Petri nets),
\SpaceEx{}~\cite{FLDCRLRGDM11} (which does not perform parameter synthesis),
or
\uppaal{}~\cite{LPY97} (which cannot be used here, as our algorithm requires timing parameters, not supported by \uppaal{}).

We ran experiments using \imitator{} 2.10.4 ``Butter Jellyfish'' (build 2477 \texttt{HEAD/5b53333})
on a Dell XPS 13 9360 equipped with an Intel\textregistered{} Core\texttrademark{} i7-7500U CPU @ 2.70GHz with 8\,GiB memory running Linux Mint 18.3 64\,bits.\footnote{%
	Sources, models and results are available at
	\href{https://doi.org/10.5281/zenodo.3251141}{\nolinkurl{doi.org/10.5281/zenodo.3251141}}
	and
		\href{https://www.imitator.fr/static/ATVA19/}{\nolinkurl{imitator.fr/static/ATVA19/}}.
}

\subsection{Translating programs into PTAs}\label{ss:Java2PTA}

We will consider case studies from the PTA community and from previous works focusing on privacy using (parametric) timed automata.
In addition, we will be interested in analyzing programs too.
In order to apply our method to the analysis of programs, we need a systematic way of translating a program (\eg{} a Java program) into a PTA.
In general, precisely modeling the execution time of a program using models like timed automata is highly non-trivial due to complication of hardware pipelining, caching, OS scheduling, etc.
The readers are referred to the rich literature in, for instance, \cite{LYGY10}.
In this work, we instead make the following simplistic assumption on execution time of a program statement and focus on solving the parameter synthesis problem.
How to precisely model the execution time of programs is orthogonal and complementary to our work.

We assume that the execution time of a program statement other than \stylecode{Thread.sleep(n)} is within a range $[0,\epsilon]$ where $\epsilon$ is a small integer constant (in milliseconds), whereas the execution time of statement \stylecode{Thread.sleep(n)} is within a range $[n , n+\epsilon]$.
In fact, we choose to keep $\epsilon$ \emph{parametric} to be as general as possible, and to not depend on particular architectures.

Our test subject is a set of benchmark programs from the DARPA Space/Time Analysis for Cybersecurity (STAC) program.\footnote{\url{https://github.com/Apogee-Research/STAC/}}
	These programs are being released publicly to facilitate researchers to develop methods and tools for identifying STAC vulnerabilities in the programs.
	These programs are simple yet non-trivial, and were built on purpose to highlight vulnerabilities that can be easily missed by existing security analysis tools.

\subsection{A richer framework}\label{ss:richer}

The symbolic representation of variables and parameters in \imitator{} allows us to reason \emph{symbolically} concerning variables.
That is, instead of enumerating all possible (bounded) values of \styledisc{x} and \styledisc{secret} in \cref{figure:example-Java:PTA}, we turn them to parameters (\ie{} unknown constants), and \imitator{} performs a symbolic reasoning.
Even better, the analysis terminates for this example even when no bound is provided on these variables.
This is often not possible in (non-parametric) timed automata based model checkers, which usually have to enumerate these values.
Therefore, in our PTA representation of Java programs, we turn all user-input variable and secret constant variables to non-timing rational-valued parameters, also supported by \imitator{}.
Other local variables are implemented using \imitator{} discrete (shared, global) variables.

We also discuss how to enlarge the scope of our framework.

\paragraph{Multiple private locations}
This can be easily achieved by setting $\bflag$ to $\BTrue$ along any incoming transition of one of these private locations.

\paragraph{Multiple final locations}
The technique used depends on whether these multiple final locations can be distinguished or not.
If they are indistinguishable (\ie{} the observer knows when the program has terminated, but not in which state), then it suffices to merge all these final locations in a single one, and our framework trivially applies.
If they are distinguishable, then one analysis needs to be conducted on each of these locations (with a different parameter~$\paramabs$ for each of these), and the obtained constraints must be intersected.

\paragraph{Access to high-level variables}
In the literature, a distinction is sometimes made between low-level (``public'') and high-level (``private'') variables.
Opacity or non-interference can be defined in terms of the ability for an observer to deduce some information on the high-level variables.

\begin{example}
	For example, in \cref{figure:example:VNN18} (where $\styleclock{cl}$ is a clock and $\styledisc{h}$ a variable), if $\loc_2$ is reachable in 20 time units, then it is clear that the value of the high-level variable $\styledisc{h}$ is negative.
\end{example}

Our framework can also be used to address this problem, \eg{} by setting $\bflag$ to~$\BTrue$, not on locations but on selected tests / valuations of such variables.

\begin{example}
	For example, setting $\bflag$ to~$\BTrue$ on the upper transition from~$\loc_1$ to~$\loc_2$ in \cref{figure:example:VNN18}, the answer to the timed opacity computation problem is $\Times = (30, \infty)$, and the system is therefore not opaque since $\loc_2$ can be reached for any execution time in~$[0, \infty)$.
\end{example}
\begin{figure}[tb]
	\centering
	\footnotesize
	\begin{tikzpicture}[scale=1, auto, ->, >=stealth']

		\node[location, initial] at (0,0) (l1) {$\loc_1$};

		\node[location, final] at (2, 0) (l2) {$\loc_2$};

		\path (l1) edge[bend left] node[align=center]{$\styledisc{h} > 0$ \\ $\styleclock{cl} > 30$} (l2);

		\path (l1) edge[bend right] node[below,align=center]{$\styledisc{h} \leq 0$} (l2);

	\end{tikzpicture}
	\caption{\cite[Fig.~5]{VNN18}}
	\label{figure:example:VNN18}

\end{figure}
\subsection{Experiments}\label{ss:experiments}
\subsubsection{Benchmarks}\label{sss:benchmarks}

As a proof of concept, we applied our method to a set of examples from the literature.
The first five models come from previous works from the literature~\cite{GMR07,BCLR15,VNN18}, also addressing non-interference or opacity in timed automata.\LongVersion{\footnote{%
	As most previous works on opacity and timed automata do not come with an implementation nor with benchmarks, it is not easy to find larger models coming in the form of~TAs.
}}
In addition, we used two common models from the (P)TA literature, not necessarily linked to security:
	a toy coffee machine (\stylebench{Coffee}) used as benchmark in a number of papers,
	and
	a model Fischer's mutual exclusion protocol (\stylebench{Fischer-HRSV02}) \cite{HRSV02}.
In both cases, we added manually a definition of private location (the number of sugars ordered, and the identity of the process entering the critical section, respectively), and we verified whether they are opaque \wrt{} these internal behaviors.

We also applied our approach to a set of Java programs from the aforementioned STAC library.
We use identifiers of the form \stylebench{STAC:1:n} where \stylebench{1} denotes the identifier in the library, while \stylebench{n} (resp.~\stylebench{v}) denotes non-vulnerable (resp.\ vulnerable).
We manually translated these programs to PTAs, following the method described in \cref{ss:Java2PTA}.
We used a representative set of programs from the library; however, some of them were too complex to fit in our framework, notably when the timing leaks come from calls to external libraries (\stylebench{STAC:15:v}), when dealing with complex computations such as operations on matrices (\stylebench{STAC:16:v}) or when handling probabilities (\stylebench{STAC:18:v}).
Proposing efficient and accurate ways to represent arbitrary programs into (parametric) timed automata is orthogonal to our work, and is the object of future works.

\subsubsection{Timed opacity computation}\label{sss:timed-opacity-computation}

First, we \emph{verified} whether a given TA model is opaque, \ie{} if for all execution times reaching a given final location, both an execution goes through a given private location and an execution does not go through this private location.
To this end, we also answer the timed opacity computation problem, \ie{} to synthesize all execution times for which the system is opaque.
While this problem can be verified on the region graph (\cref{proposition:ET-opacity-computation}), we use the same framework as in \cref{section:synthesis}, but without parameters in the original TA.
That is, we use the Boolean flag~$\bflag$ and the parameter $\paramabs$ to compute all possible execution times.
In other words, we use a parametric analysis to solve a non-parametric problem.

\newcommand{\columnRef}[1]{}

\begin{table}[tb]
	\caption{Experiments: timed opacity}
	\centering
	\footnotesize
	\ShortVersion{\scriptsize}

	\setlength{\tabcolsep}{2pt} %

		\begin{tabular}{| c | c | c | c | c | c | r | r | c |} %
		\hline
		\LongVersionTable{\multicolumn{3}{| c |}{\cellHeader{}Model} & \multicolumn{3}{ c |}{\cellHeader{}Transf. PTA} & \multicolumn{3}{c |}{\cellHeader{}Result}}
		\ShortVersionTable{\multicolumn{3}{| c |}{\cellHeader{}Model} & \multicolumn{3}{ c |}{\cellHeader{}Transf. PTA} & \multicolumn{2}{c |}{\cellHeader{}Result}}
		\\
		\hline
		\rowHeader{}
		Name & \columnRef{Reference & }$|\A|$ & $|\Clock|$ & $|\A|$ & $|\Clock|$ & $|\Param|$\LongVersionTable{ & States} &Time (s) & Opaque? \\
		\hline
		\cite[Fig.~5]{VNN18} & \columnRef{\cite{VNN18} &} 1 & 1 & 2 & 3 & 3 \LongVersionTable{ & 13} & 0.02 & \cellFixable{}  \\ %
		\hline
		\cite[Fig.~1b]{GMR07} & \columnRef{\cite{GMR07} & }1 & 1 & 2 & 3 & 1 \LongVersionTable{& 25} & 0.04 & \cellFixable{} \\ %
		\hline
		\cite[Fig.~2a]{GMR07} & \columnRef{\cite{GMR07} & }1 & 1 & 2 & 3 & 1 \LongVersionTable{& 41} & 0.05 & \cellFixable{} \\ %
		\hline
		\cite[Fig.~2b]{GMR07} & \columnRef{\cite{GMR07} & }1 & 1 & 2 & 3 & 1 \LongVersionTable{& 41} & 0.02 & \cellFixable{} \\ %
		\hline
		Web privacy problem \cite{BCLR15} & \columnRef{\cite{BCLR15} & }1 & 2 & 2 & 4 & 1 \LongVersionTable{& 105} & 0.07 & \cellFixable{} \\ %
		\hline
		\stylebench{Coffee} & \columnRef{\cite{AMP21} & }1 & 2 & 2 & 5 & 1 \LongVersionTable{& 43} & 0.05 & \cellYes{} \\ %
		\hline
		\stylebench{Fischer-HSRV02} & \columnRef{\cite{HRSV02} & }3 & 2 & 6 & 5 & 1 \LongVersionTable{& 2495} & 5.83 & \cellFixable{} \\ %
		\hline
		\stylebench{STAC:1:n} & \columnRef{ & }\multicolumn{2}{c |}{\nbLoC{69}} & 2 & 3 & 6 \LongVersionTable{& 65} & 0.12 & \cellFixable{} \\ %
		\hline
		\stylebench{STAC:1:v} & \columnRef{ & }\multicolumn{2}{c |}{\nbLoC{69}} & 2 & 3 & 6 \LongVersionTable{& 63} & 0.11 & \cellNo{} \\ %
		\hline
		\stylebench{STAC:3:n} & \columnRef{ & }\multicolumn{2}{c |}{\nbLoC{87}} & 2 & 3 & 8 \LongVersionTable{& 289} & 0.72 & \cellYes{} \\ %
		\hline
		\stylebench{STAC:3:v} & \columnRef{ & }\multicolumn{2}{c |}{\nbLoC{87}} & 2 & 3 & 8 \LongVersionTable{& 287} & 0.74 & \cellFixable{} \\ %
		\hline
		\stylebench{STAC:4:n} & \columnRef{ & }\multicolumn{2}{c |}{\nbLoC{112}} & 2 & 3 & 8 \LongVersionTable{& 904} & 6.40 & \cellNo{} \\ %
		\hline
		\stylebench{STAC:4:v} & \columnRef{ & }\multicolumn{2}{c |}{\nbLoC{110}} & 2 & 3 & 8 \LongVersionTable{& 19183} & 265.52 & \cellNo{} \\ %
		\hline
		\stylebench{STAC:5:n} & \columnRef{ & }\multicolumn{2}{c |}{\nbLoC{115}} & 2 & 3 & 6 \LongVersionTable{& 144} & 0.24 & \cellYes{} \\ %
		\hline
		\stylebench{STAC:11A:v} & \columnRef{ & }\multicolumn{2}{c |}{\nbLoC{81}} & 2 & 3 & 8 \LongVersionTable{& 5037} & 47.77 & \cellFixable{} \\ %
		\hline
		\stylebench{STAC:11B:v} & \columnRef{ & }\multicolumn{2}{c |}{\nbLoC{85}} & 2 & 3 & 8 \LongVersionTable{& 5486} & 59.35 & \cellFixable{} \\ %
		\hline
		\stylebench{STAC:12c:v} & \columnRef{ & }\multicolumn{2}{c |}{\nbLoC{81}} & 2 & 3 & 8 \LongVersionTable{& 1177} & 18.44 & \cellNo{} \\ %
		\hline
		\stylebench{STAC:12e:n} & \columnRef{ & }\multicolumn{2}{c |}{\nbLoC{96}} & 2 & 3 & 8 \LongVersionTable{& 169} & 0.58 & \cellNo{} \\ %
		\hline
		\stylebench{STAC:12e:v} & \columnRef{ & }\multicolumn{2}{c |}{\nbLoC{85}} & 2 & 3 & 8 \LongVersionTable{& 244} & 1.10 & \cellFixable{} \\ %
		\hline
		\stylebench{STAC:14:n} & \columnRef{ & }\multicolumn{2}{c |}{\nbLoC{88}} & 2 & 3 & 8 \LongVersionTable{& 1223} & 22.34 & \cellFixable{} \\ %
		\hline
	\end{tabular}

	\label{table:nonparametric}
\end{table}

We tabulate the experiments results in \cref{table:nonparametric}.\label{newtext:headers}
We give from left to right the model name, the numbers of automata and of clocks in the original timed automaton (this information is not relevant for Java programs as the original model is not a TA), the numbers of automata\footnote{%
	As usual, it may be simpler to write PTA models as a network of PTAs.
	Recall from \cref{definition:parallel} that a network of PTAs gives a PTA.
	In this case, $|\A|$ denotes the number of input PTA components.
	\label{footnote:network}
}, of clocks and of parameters in the transformed PTA, the computation time in seconds (for the timed opacity computation problem), and the result.
In the result column, ``\cellYes{}'' (resp.~``\cellNo{}'') denotes that the model is opaque (resp.\ is not opaque, \ie{} vulnerable), while ``\cellFixable{}'' denotes that the model is not opaque, but could be fixed.
That is, although $\PrivDurReach{\valuate{\A}{\pval}}{\locpriv}{\locfinal} \neq \PubDurReach{\valuate{\A}{\pval}}{\locpriv}{\locfinal}$, their intersection is non-empty and therefore, by tuning the computation time, it may be possible to make the system opaque.
This will be discussed in \cref{ss:rendering-opaque}.

Even though we are interested here in timed opacity computation (and not in synthesis), note that all models derived from Java programs feature the parameter~$\styleparam{\epsilon}$.
The result is obtained by variable elimination, \ie{} by existential quantification over the parameters different from~$\paramabs$.
In addition, the number of parameters is increased by the parameters encoding the symbolic variables (such as $\styledisc{x}$ and $\styledisc{secret}$ in \cref{figure:example-Java:PTA}).

Concerning the Java programs, we decided to keep the most abstract representation, by imposing that each instruction lasts for a time in~$[0,\styleparam{\epsilon}]$, with $\styleparam{\epsilon}$ a parameter.
However,
	fixing an identical (parametric) time $\styleparam{\epsilon}$ for all instructions,
	or fixing an arbitrary time in a constant interval $[0, \epsilon]$ (for some constant $\epsilon$, \eg{}~1),
	or even fixing an identical (constant) time $\epsilon$ (\eg{}~1) for all instructions,
significantly speeds up the analysis.
These choices can be made for larger models.

\paragraph*{Discussion}
Overall, our method is able to answer the timed opacity computation problem for practical case studies, exhibiting which execution times are opaque (timed opacity computation problem), and whether \emph{all} execution times indeed guarantee opacity (timed opacity problem).
	\label{newtext:practical}

In many cases, while the system is not opaque, we are able to \emph{infer} the execution times guaranteeing opacity (cells marked ``\cellFixable{}'').
This is an advantage of our method \wrt{} methods outputting only binary answers.

We observed some mismatches in the Java programs, \ie{} some of the programs marked \stylebench{n} (non-vulnerable) in the library are actually vulnerable according to our method.
This mainly comes from the fact that the STAC library expect tools to use imprecise analyses on the execution times, while we use an exact method.
Therefore, a very small mismatch between $\PrivDurReach{\valuate{\A}{\pval}}{\locpriv}{\locfinal}$ and $\PubDurReach{\valuate{\A}{\pval}}{\locpriv}{\locfinal}$ will lead our algorithm to answer ``not opaque'', while some methods may not be able to differentiate this mismatch from imprecision (noise).
This is notably the case of \stylebench{STAC:14:n} where some action lasts either 5,010,000 or 5,000,000 time units depending on some secret, which our method detects to be different, while the library does not.
For \stylebench{STAC:1:n}, using our data, the difference in the execution time upper bound between an execution performing some secret action and an execution not performing it is larger than~1\,\%, which we believe is a value which is not negligible, and therefore this case study might be considered as vulnerable.
\label{oldtext:stac1n}

\stylebench{STAC:4:n} requires a more detailed discussion.
	This particular program is targeting vulnerabilities that can be detected easily \emph{when they accumulate}, typically in loops.
	This program checks a number of times (10) a user-input password, and each password check is made in the most insecure way, \ie{} by returning ``incorrect'' as soon as one character differs between the input password and the expected password.
	This way is very insecure because the execution time is proportional to the number of consecutive correct characters in the input password and, by observing the execution time, an attacker can guess how many characters are correct, and therefore using a limited number of tests, (s)he will eventually guess the correct password.
	The difference between the vulnerable (\stylebench{STAC:4:v}) and the non-vulnerable (\stylebench{STAC:4:n}) versions is that the non-vulnerable version immediately stops if the password is incorrect, and performs the 10 checks only if the password is correct.
	Therefore, while the computation time is very different between the correct input password and any incorrect input password, it is however very similar between an incorrect input password that would only be incorrect because, say, of the last character (\eg{} ``\texttt{kouignamaz}'' while the expected password is ``\texttt{kouignaman}''), and a completely incorrect input password differing as early as the first character (\eg{} ``\texttt{andouille}'').
	This makes the attacker's task very difficult.
	The main reason for the STAC library to label \stylebench{STAC:4:n} as a non-vulnerable program is because of the ``very similar'' nature of the computation times between an incorrect input password that would only be incorrect because of the last character, and a completely incorrect input password.
	(In contrast, the vulnerable version \stylebench{STAC:4:v} is completely vulnerable because this time difference is amplified by the loop, here 10~times.)
	While ``very similar'' might be acceptable for most tools, in our setting based on formal verification, we \emph{do} detect that testing ``\texttt{kouignamaz}'' or testing ``\texttt{kouignamzz}'' will yield a slightly faster computation time for the second input, because the first incorrect letter occurs earlier---and the program is therefore vulnerable.
	\label{newtext:mismatch4}

\subsubsection{Timed opacity synthesis}\label{sss:tos}

Then, we address the timed opacity synthesis problem.
In this case, we \emph{synthesize} both the execution time and the internal values of the parameters for which one cannot deduce private information from the execution time.

We consider the same case studies as for timed opacity computation; however, the Java programs feature no internal ``parameter'' and cannot be used here.
Still, as a proof of concept, we artificially enriched one of them (\stylebench{STAC:3:v}) as follows: in addition to the parametric value of $\styleparam{\epsilon}$ and the execution time, we parameterized one of the \stylecode{sleep} timers.
The resulting constraint can help designers to refine this latter value to ensure opacity.
Note that it may not be that easy to tune a Java program to make it non-opaque: while this is reasonably easy on the PTA level (restraining the execution times using an additional clock), this may not be clear on the original model: Making a program terminate slower than originally is easy with a \texttt{Sleep} statement; but making it terminate ``earlier'' is less obvious, as it may mean an abrupt termination, possibly leading to wrong results.

We tabulate the results in \cref{table:parametric}, where the columns are similar to \cref{table:nonparametric}.
A difference is that the first $|\Param|$ column denotes the number of parameters in the original model (without counting these added by our transformation).
In addition, \cref{table:parametric} does not contain a ``vulnerable?'' column as we \emph{synthesize} the condition for which the model is non-vulnerable, and therefore the answer is non-binary.
However, in the last column (``Constraint''), we make explicit whether no valuations ensure opacity (``\cellKnone{}''), all of them (``\cellKall{}''), or some of them (``\cellKsome{}'').

\paragraph*{Discussion}
An interesting outcome is that the computation time is comparable to the (non-parametric) timed opacity computation, with an increase of up to~20\,\% only.
In addition, for all case studies, we exhibit at least some valuations for which the system can be made opaque.
Also note that our method always terminates for these models, and therefore the result exhibited is complete.
Interestingly, \stylebench{Coffee} is opaque for any valuation of the 3~internal parameters.

\begin{table}[tb]
	\caption{Experiments: timed opacity synthesis}
	\centering
	\footnotesize
	\ShortVersion{\scriptsize}

	\setlength{\tabcolsep}{2pt} %
	\LongVersionTable{\begin{tabular}{| c | c | c | c | c | c | c | r | r | c |}} %
	\ShortVersionTable{\begin{tabular}{| c | c | c | c | c | c | r | r | c |}} %
		\hline
		\LongVersionTable{\multicolumn{4}{| c |}{\cellHeader{}Model} & \multicolumn{3}{ c |}{\cellHeader{}Transf. PTA} & \multicolumn{3}{c |}{\cellHeader{}Result}}
		\ShortVersionTable{\multicolumn{4}{| c |}{\cellHeader{}Model} & \multicolumn{3}{ c |}{\cellHeader{}Transf. PTA} & \multicolumn{2}{c |}{\cellHeader{}Result}}
		\\
		\hline
		\rowHeader{}
		Name & \columnRef{Reference & }$|\A|$ & $|\Clock|$ & $|\Param|$ & $|\A|$ & $|\Clock|$ & $|\Param|$ \LongVersionTable{& States} & Time (s) & Constraint \\
		\hline
		\cite[Fig.~5]{VNN18} & \columnRef{\cite{VNN18} &} 1 & 1 & 0 & 2 & 3 & 4 \LongVersionTable{& 13} & 0.02 & \cellKsome{}\\
		\hline
		\cite[Fig.~1b]{GMR07} & \columnRef{\cite{GMR07} & }1 & 1 & 0 & 2 & 3 & 3 \LongVersionTable{& 25} & 0.03 & \cellKsome{}\\
		\hline
		\cite[Fig.~2]{GMR07} & \columnRef{\cite{GMR07} & }1 & 1 & 0 & 2 & 3 & 3 \LongVersionTable{& 41} & 0.05 & \cellKsome{}\\
		\hline
		Web privacy problem \cite{BCLR15} & \columnRef{\cite{BCLR15} & }1 & 2 & 2 & 2 & 4 & 3 \LongVersionTable{& 105} & 0.07 & \cellKsome{}\\
		\hline
		\stylebench{Coffee} & \columnRef{\cite{AMP21} & }1 & 2 & 3 & 2 & 5 & 4 \LongVersionTable{& 85} & 0.10 & \cellKall{}\\
		\hline
		\stylebench{Fischer-HSRV02} & \columnRef{\cite{HRSV02} & }3 & 2 & 2 & 6 & 5 & 3 \LongVersionTable{& 2495} & 7.53 & \cellKsome{}\\
		\hline
		\stylebench{STAC:3:v} & \columnRef{ & }\multicolumn{2}{c |}{\nbLoC{87}} & 2 & 2 & 3 & 9 \LongVersionTable{& 361} & 0.93 & \cellKsome{}\\
		\hline
	\end{tabular}

	\label{table:parametric}
\end{table}
\subsection{``Repairing'' a non-opaque PTA}\label{ss:rendering-opaque}

Our method gives a result in time of a union of polyhedra over the internal timing parameters and the execution time.
On the one hand, we believe tuning the internal timing parameters should be easy: for a program, an internal timing parameter can be the duration of a \stylecode{sleep}, for example.
On the other hand, tuning the execution time of a program may be more subtle.
A solution is to enforce a minimal execution time by adding a second thread in parallel with a \stylecode{Wait()} primitive to ensure a minimal execution time.
Ensuring a \emph{maximal} execution time can be achieved with an exception stopping the program after a given time; however there is a priori no guarantee that the result of the computation is correct.
\section{Conclusion}\label{section:conclusion}

In this work, we proposed an approach based on parametric timed model checking to not only decide whether the model of a timed system can be subject to timing information leakage, but also to \emph{synthesize} internal timing parameters and execution times that render the system opaque.
We implemented our approach in a framework based on \imitator{},
and performed experiments on case studies from the literature and from a library of Java programs.

We now discuss future works in the following.

\paragraph*{Theory}
We proved decidability of the timed opacity computation problem (\cref{proposition:ET-opacity-computation}) and of the full timed opacity decision problem (\cref{proposition:ET-full-opacity-decsion}) for TAs, but we only provided an upper bound (3EXPTIME) on the complexity.
It can be easily shown that these problems are at least PSPACE, but the exact complexity remains to be exhibited.

In addition, the decidability of several ``low-dimensional'' problems (\ie{} with ``small'' number of clocks or parameters) remains open.
Among these, the one-clock case for parametric timed opacity emptiness (\cref{proposition:TOE-undecidability}) remains open:
	that is, is the timed opacity emptiness problem decidable for PTAs using a single clock?
Our method in \cref{section:synthesis} consists in duplicating the automaton and adding a clock that is never reset, thus resulting in a PTA with 3 clocks, for which reachability-emptiness is undecidable~\cite{AHV93}.
However, since one of the clocks is never reset, and since the automaton is structurally constrained (it is the result of the composition of two copies of the same automaton), decidability might be envisioned.
Recall that the 2-clock reachability-emptiness problem is a famous open problem~\cite{Andre19STTT}, despite recent advances, notably over discrete time~\cite{BO14,GH21}.
The 1-clock question also remains open for full timed opacity emptiness (\cref{theorem:FTOE-undecidability}).
The minimum number of parameters required for our proof of the undecidability of the full timed opacity emptiness problem for PTAs (resp.\ L/U-PTAs) to work is~2 (resp.~4), as seen in \cref{theorem:FTOE-undecidability} (resp.\ \cref{theorem:FTOE:LU}); it is open whether using less parameters can render these problems decidable.

Finally, concerning L/U-PTAs, we proved two negative results, despite the decidability of the timed opacity emptiness problem (\cref{proposition:decidability-LU}): the undecidability of the full timed opacity emptiness (\cref{theorem:FTOE:LU}) and the intractability of timed opacity synthesis (\cref{proposition:intractabilitySynthTO}).
It remains open whether these results still apply to the more restrictive class of U-PTAs~\cite{BlT09}.

\paragraph*{Full timed opacity synthesis}

We leave full timed opacity synthesis as future work; while we could certainly reuse partially our algorithm, this is not entirely trivial, as we need to select only the parameter valuations for which the whole set of execution times is exactly the set of opaque times.

\paragraph*{Applications}

The translation of the STAC library required some non-trivial creativity: while the translation from programs to quantitative extensions of automata is orthogonal to our work, proposing automated translations of (possibly annotated) programs to timed automata dedicated to timing analysis is on our agenda.

	Adding probabilities to our framework will be interesting, helping to quantify the execution times of ``untimed'' instructions in program with a finer grain than an interval; also note that some benchmarks make use of probabilities (notably \stylebench{STAC:18:v}).

	Finally, \imitator{} is a general model checker, not specifically aimed at solving the problem we address here.
	Notably, constraints managed by PPL contain all variables (clocks, timing parameters, and parameters encoding symbolic variables of programs), yielding an exponential complexity.
	Separating certain types of independent variables (typically parameters encoding symbolic variables of programs, and other variables) should increase efficiency.
\ifdefined\VersionAuthorforArXiV
\else
	\begin{acks}
		\ouracks{}
	\end{acks}
\fi
	\newcommand{\CCIS}{Communications in Computer and Information Science}
	\newcommand{\ENTCS}{Electronic Notes in Theoretical Computer Science}
	\newcommand{\FI}{Fundamenta Informormaticae}
	\newcommand{\FMSD}{Formal Methods in System Design}
	\newcommand{\IJFCS}{International Journal of Foundations of Computer Science}
	\newcommand{\IJSSE}{International Journal of Secure Software Engineering}
	\newcommand{\IPL}{Information Processing Letters}
	\newcommand{\JLAP}{Journal of Logic and Algebraic Programming}
	\newcommand{\JLC}{Journal of Logic and Computation}
	\newcommand{\LMCS}{Logical Methods in Computer Science}
	\newcommand{\LNCS}{Lecture Notes in Computer Science}
	\newcommand{\RESS}{Reliability Engineering \& System Safety}
	\newcommand{\STTT}{International Journal on Software Tools for Technology Transfer}
	\newcommand{\TCS}{Theoretical Computer Science}
	\newcommand{\ToPNoC}{Transactions on Petri Nets and Other Models of Concurrency}
	\newcommand{\TSE}{IEEE Transactions on Software Engineering}

\ifdefined\VersionAuthorforArXiV
	\renewcommand*{\bibfont}{\footnotesize}
	\printbibliography %
\else
	\printbibliography
\fi
\newpage
\appendix

\section{The code of the Java example}\label{appendix:Java}

\begin{lstlisting}[language=Java]
import java.io.BufferedReader;
import java.io.IOException;
import java.io.InputStreamReader;
import java.io.PrintWriter;
import java.net.ServerSocket;
import java.net.Socket;

//Coefficients are Disregarded
public class Category1_vulnerable {
    private static final int port = 8000;
    private static final int secret = 1234;
    private static final int n = 32;
    private static ServerSocket server;

    private static void checkSecret(int guess) throws InterruptedException {
        if (guess <= secret) {
            for (int i = 0; i < n; i++) {
                for (int t = 0; t < n; t++) {
                    Thread.sleep(1);
                }
            }
        } else {
            for (int i = 0; i < n; i++) {
                for (int t = 0; t < n; t++) {
                    Thread.sleep(2);
                }
            }
        }
    }

    private static void startServer() {
        try {
            server = new ServerSocket(port);
            System.out.println("Server Started Port: " + port);
            Socket client;
            PrintWriter out;
            BufferedReader in;
            String userInput;
            int guess;
            while (true) {
                client = server.accept();
                out = new PrintWriter(client.getOutputStream(), true);
                in = new BufferedReader(new InputStreamReader(client.getInputStream()));

                userInput = in.readLine();
                try {
                    guess = Integer.parseInt(userInput);
                    if(guess < 0) {
                        throw new IllegalArgumentException();
                    }
                    checkSecret(guess);
                    out.println("Process Complete");
                } catch (IllegalArgumentException | InterruptedException e) {
                    out.println("Unable to Process Input");
                }
                client.shutdownOutput();
                client.shutdownInput();
                client.close();
            }
        } catch (IOException e) {
            System.exit(-1);
        }
    }

    public static void main(String[] args) throws InterruptedException {
        startServer();
    }
}
\end{lstlisting}

\label{newtext:2*32*32}%
Note that the two ``\texttt{for}'' loops featuring a \texttt{Thread.sleep(1}) (resp.~\texttt{2}) could be equivalently replaced with a simple \texttt{Thread.sleep(32*32)} (resp.\ \texttt{Thread.sleep(2*32*32)}) statement, but \begin{oneenumeration}
	\item this is the way the program is presented in the DARPA library, and 
	\item a (minor) difficulty may come from these loops instead of a simple \texttt{Thread.sleep(32*32)} statement.
\end{oneenumeration}

\ifdefined\WithReply
	\clearpage
	\newpage
	\input{letter2.tex}
\fi

\end{document}